\documentclass[a4paper,11pt,english]{article}

\usepackage{babel}
\usepackage{amsmath,amsfonts,amssymb,amsthm,yhmath,mathrsfs, amsbsy}
\usepackage{latexsym}
\usepackage{mathtools}
\usepackage[utf8]{inputenc}
\usepackage{fancyhdr}
\usepackage[nottoc]{tocbibind}
\usepackage{indentfirst}
\usepackage{paralist}
\usepackage[a-1b]{pdfx}
\usepackage[pdfa]{hyperref}
\usepackage{graphicx}
\usepackage{dcpic, pictex}
\usepackage{etoolbox}
\usepackage{tikz}
\usepackage{enumerate}
\usepackage{makeidx}
\usepackage{multicol}
\usepackage{comment}

\newtheorem{theorem}{Theorem}[section]
\newtheorem{proposition}[theorem]{Proposition}
\newtheorem{lemma}[theorem]{Lemma}
\newtheorem{corollary}[theorem]{Corollary}

\theoremstyle{definition}
\newtheorem{definition}[theorem]{Definition}

\newtheorem{remark}[theorem]{Remark}
\newtheorem{example}[theorem]{Example}
\newtheorem{notation}[theorem]{Notation}

\newcommand{\N}{\mathbb N}

\newcommand{\Z}{\mathbb Z}

\newcommand{\F}{\mathbb F}
\newcommand{\Cc}{\mathcal C}
\newcommand{\Mat}{\text{Mat}}

\newcommand\eli[1]{{{\textcolor{red}{#1}}}}
\newcommand{\MM}{\mathbb M}
\DeclareMathOperator{\rk}{rk}
\DeclareMathOperator{\srk}{srk}

\newcommand{\Cd}{\mathcal{C}^\perp}
\newcommand{\A}{\mathcal{A}}
\newcommand{\V}{\mathcal{V}}
\newcommand{\B}{\mathcal{B}}

\newcommand{\matsumi}{\mathbb M}
\DeclareMathOperator{\maxwt}{maxwt}
\DeclareMathOperator{\maxrk}{maxrk}
\DeclareMathOperator{\maxsrk}{maxsrk}
\DeclareMathOperator{\GL}{GL}

\title{Optimal anticodes, MSRD codes, and generalized weights in the sum-rank metric}
\author{E.~Camps Moreno, E.~Gorla, C.~Landolina, E.~Lorenzo Garc\'ia,\\ U.~Mart\'inez-Pe\~{n}as, F.~Salizzoni}
\date{}

\begin{document}
\maketitle

\begin{abstract}
Sum-rank metric codes have recently attracted the attention of many researchers, due to their relevance in several applications.
Mathematically, the sum-rank metric is a natural generalization of both the Hamming metric and the rank metric. In this paper, we provide an Anticode Bound for the sum-rank metric, which extends the corresponding Hamming and rank-metric Anticode bounds. We classify then optimal anticodes, i.e., codes attaining the sum-rank metric Anticode Bound. We use these optimal anticodes to define generalized sum-rank weights and we study their main properties. In particular, we prove that the generalized weights of an MSRD code are determined by its parameters. As an application, in the Appendix we explain how generalized weights measure information leakage in multishot network coding. 
\end{abstract}

\section*{Introduction}

The sum-rank metric has recently attracted attention in Coding Theory due to its applications in reliable and secure multishot network coding \cite{multishot, secure-multishot}, rate-diversity optimal space-time codes \cite{spacetime-kumar, mohannad}, and PMDS codes for repair in distributed storage \cite{cai}, among others. 
Furthermore, the sum-rank metric is a natural generalization of both the Hamming metric and the rank metric, thus providing a common theoretical framework for these two well-studied metrics.

Several constructions of sum-rank metric codes exist in the literature. The first constructions were mainly of \textit{convolutional codes}, see \cite{cost-action-multishot} for a survey and references. In this manuscript, we consider \textit{block codes}. A trivial Singleton Bound on their minimum sum-rank distance may be immediately derived from the classical Singleton Bound on minimum Hamming distance \cite[Prop. 34]{linearizedRS}. Any code attaining the Singleton Bound for the rank metric (i.e., any maximum rank-distance (MRD) code, including Gabidulin codes \cite{Del, gabidulin, roth}) also attains it for the sum-rank metric, that is, it is also a maximum sum-rank distance (MSRD) code. However, the parameters of MRD codes (including the matrix sizes) are very strongly restricted. Furthermore, their decoding algorithms are over finite fields whose sizes are exponential in the code length (i.e., the total number of columns), making such decoding algorithms slow for large parameters. What makes the study of MSRD codes interesting is that there are MSRD codes not coming from MRD codes and attaining a wider range of parameters, including codes \cite{linearizedRS} with decoding algorithms over finite fields of sub-exponential size \cite{secure-multishot}. 

Since then, other families of sum-rank metric codes have been found and studied \cite{maire, interleaved, twisted, BGROAC, BGRMSRD}. 
However, previous works, with the exception of \cite{BGROAC, BGRMSRD}, consider sum-rank metric codes where the number of columns and/or rows are equal at different positions. A general Singleton Bound for arbitrary numbers of columns and rows was given in \cite[Th. 3.2]{BGRMSRD}, together with corresponding MSRD codes for certain parameter ranges \cite[Sec. 7]{BGRMSRD}.

In the context of wire-tap channels of type II, Wei introduced generalized Hamming weights \cite{Wei}, which measure information leakage to an undesired wire-tapper. Generalized Hamming weights also constitute a Hamming-metric invariant of a code, and thus they are a useful tool in the classification of Hamming-metric codes. Such weights were extended in \cite{rgrw} to generalized rank weights of vector codes linear over an extension field. Such weights measure information leakage to a wire-tapper in singleshot linear network coding. Similarly, generalized sum-rank weights for vector codes may be obtained \cite{gsrws}. In \cite{Rav16, rgmw}, two extensions of generalized rank weights were given for matrix codes (thus only linear over the base field of the network). The two definitions differ in terms of their applications. The first ones \cite{Rav16}, called Delsarte generalized weights, constitute a rank-metric invariant of the codes, whereas the second ones \cite{rgmw}, called generalized matrix weights, measure information leakage to a wire-tapper but are not rank-metric invariants. Unfortunately, this discrepancy may not be saved, as there is only one possible definition of generalized weights of matrix codes \cite{rgmw} that measure information leakage in linear network coding, and it does not lead to rank-metric invariants. See also \cite[Sec. 5]{G21}.

In this work, we introduce generalized sum-rank weights of codes which are only linear over the base field, and which we think of as matrix codes. We will focus on a definition that extends Delsarte generalized weights \cite{Rav16}, and in the Appendix, we show how to modify the definition in order to extend generalized matrix weights \cite{rgmw} and to measure information leakage in multishot network coding. The proofs of the main properties for the second definition can be trivially adapted from the corresponding results for the first definition of generalized weights. 

Our main definition of generalized sum-rank weights is based on optimal anticodes for the sum-rank metric, in line with the rank-metric case \cite{Rav16}. To this end, we provide in Theorem \ref{theoremab} an Anticode Bound for the sum-rank metric, which extends the Hamming-metric Anticode Bound \cite[Prop. 6]{Rav16} and the rank-metric Anticode Bound \cite[Prop. 47]{Rav16a}. We then provide in Theorem \ref{thm:OAC} a classification of optimal anticodes in the sum-rank metric, that is, codes attaining the sum-rank metric Anticode Bound. Recently in \cite{BGROAC}, a different Anticode Bound was given for the sum-rank metric. However, our bound is sharper and the resulting optimal anticodes lead to a definition of generalized sum-rank weights that satisfy desirable properties, whereas generalized weights based on anticodes as in \cite{BGROAC} do not recover the minimum sum-rank distance of the code.

The remainder of this manuscript is organized as follows. In Section \ref{sec:preliminaries}, we collect some preliminaries on the sum-rank metric. In Section \ref{sec:rank}, we study and lower bound the maximum rank of cosets of a linear rank-metric code, extending results from Meshulam \cite{Mes85} to cosets. Using these results, we provide in Section \ref{sec:OAC} our Anticode Bound for sum-rank metric codes and we provide an explicit description and classification of optimal anticodes for the sum-rank metric. In Section \ref{sec:isometries}, we study linear isometries of sum-rank metric codes. Such isometries allow us to define the notion of equivalent codes, which allows us to say if a given parameter of a code is a sum-rank invariant. In Section \ref{sec:generalizedweights}, we use optimal anticodes to define and obtain the main properties of generalized sum-rank weights. Finally in Section \ref{sec:MRD}, we use the previous results to define and study MSRD codes and $ r $-MSRD codes in the general scenario considered in this work, namely matrix codes with different numbers of rows and/or columns at different positions.

\subsection*{Acknowledgement} The authors thank Alberto Ravagnani for making us aware of the exception in Theorem 2.11 and of the related work of Cl\'ement de Seguins Pazzis. The first author was partially supported by CONACyT and by the Swiss Confederation through the Swiss Government Excellence Scholarship no. 2020.0086. The research of the fourth author was partially funded by the \textit{Melodia} ANR-20-CE40-0013 project.

\section{Preliminaries and notation}\label{sec:preliminaries}

For a prime power $q$, let $\F_q$ be the finite field with $q$ elements. For positive integers $m \geq n$, we denote by $\F_q^{m \times n}$ the set of $m \times n $ matrices with entries in $\F_q$. We denote by $\rk(M)$ the rank of a matrix $M \in \F_q^{m \times n}$ and by $\dim(V)$ the dimension of an $\F_q$-linear space $V$. We denote by $0$ the zero vector space.

For a positive integer $r$, we let $[r]$ be the set $\{1,\ldots,r\}$. 
For a matrix $M\in\F_q^{m\times n}$ and $S\subseteq [m], L\subseteq[n]$ we let $M(S,L)$ denote the submatrix of $M$ consisting of the rows indexed by $S$ and the columns indexed by $L$. For $(s,l)\in[m]\times[n]$, $M(s,l)$ denotes the entry of $M$ in position $(s,l)$.
Moreover, we let $E_{s,l}\in\F_q^{m\times n}$ be the matrix whose entries are equal to zero, except for a one in position $(s,l)\in[m]\times[n]$.

Fix positive integers $\ell , n_1, \ldots , n_\ell,m_1, \ldots ,  m_\ell$ such that $m_1\geq  \ldots \geq   m_\ell$ and 
 $ n_i\leq m_i$ for $i\in[\ell]$. 
 We write $n = n_1 + \ldots + n_\ell$ and let 
 $$\MM=\F_q^{m_1\times n_1}\times\ldots\times\F_q^{m_\ell\times n_\ell}.$$
In particular, if $\ell=1$, then $n=n_1$ and we let $m=m_1$. Then $m\geq n$ and $\MM=\F_q^{m \times n}$.
 
\begin{definition} 
Let $C = (C_1, \ldots , C_\ell) \in \MM$, where $C_i\in\F_q^{m_i\times n_i}$ for $i\in[\ell]$. We define the \textbf{sum-rank weight} of $C$ as
$$\mathrm{srk}(C) = \sum_{i=1}^\ell \rk(C_i).$$

The \textbf{sum-rank metric} is then defined as 
$$\begin{array}{ccccc}
d  &: & \MM \times \MM & \longrightarrow & \N \\
& & (C, D) & \longmapsto & \mathrm{srk}(C-D).
\end{array}$$

A  \textbf{linear sum-rank metric code} $\Cc$ is an $\F_q$-linear subspace of $\MM$ endowed with the sum-rank metric. Throughout the paper, we will refer to it simply as a (sum-rank) code. A code $\Cc\subseteq\MM$ is {\bf non-trivial} if $\Cc\neq 0,\MM$. 

The \textbf{minimum distance} of a code $0\neq\mathcal{C} \subseteq \MM$ is 
$$ d(\Cc) = \min\{ \mathrm{srk}(C) : C \in \Cc \setminus \{ 0\}  \}$$ and the \textbf{maximum sum-rank distance} is  
$$\maxsrk(\Cc) = \max\{ \mathrm{srk}(C) : C \in \Cc  \}.$$
\end{definition}

Notice that, if we let $\ell = 1$, then $\Cc\subseteq \F_q^{m_1 \times n_1}$ is a rank-metric code. We refer the interested reader to \cite{G21} for an introduction to rank-metric codes and their invariants. If $m_1 = 1 $, then $m_2=\ldots=m_\ell=1$ and $\Cc\subseteq \F_q^n$ is a linear block code endowed with the Hamming metric.  

For a square matrix $M$, let $\mathrm{tr}(M)$ denote its trace. Then
$$\begin{array}{ccccc}
\mathrm{Tr}  &: & \matsumi \times \matsumi & \longrightarrow  & \F_q \\
     & & (D,C)    & \longmapsto    &  \sum_{i=1}^\ell \mathrm{tr}(D_i C_i^t) 
\end{array}$$
is a nondegenerate bilinear form.
We define the dual of a code as the natural extension of the dual of a rank-metric code, as defined in \cite{Del}.

\begin{definition}
Let $\Cc \subseteq \matsumi$ be a code. The dual of $\Cc$ is 
$$\Cc^\perp = \{ D \in \matsumi : \mathrm{Tr}(D,C) = 0 \mbox{ for all } C \in \Cc\}.$$
\end{definition}

\section{Maximal rank in cosets of rank-metric codes}\label{sec:rank}

In this section we provide lower bounds for the maximum rank of a coset of a rank-metric code. Our strategy is inspired by that used by Meshulam in~\cite{Mes85} and extends it to cosets of a vector space. 

Let $\prec$ be the lexicographic order on $\N\times\N$ and let 
$$\begin{array}{rcl}
\phi : \F_q^{m\times n} & \to & \N\times\N\\
M & \mapsto & \min_{\prec}\{(i,j):M(i,j)\neq 0\}.
\end{array}$$
\begin{definition}For a collection $\mathcal{M}=\{M_1,\dots, M_d\}$ of matrices in  $\F_q^{m\times n}$, we define a matrix $M$ whose entry in position $(i,j)$ is
\begin{equation*}
M(i,j)=\begin{cases}
1 & \text{if }(i,j)=\phi(M_k)\text{ for some } k\in[d],\\
0 & \text{otherwise}.
\end{cases}
\end{equation*}
Denote by $\rho(\mathcal{M})$ the minimal number of lines in $M$ which cover all ones in $M$, where a line of a matrix is either a row or a column.
\end{definition}

A set of positions $\{(i_1,j_1),\ldots,(i_r,j_r)\}$ of entries in a matrix is independent if for all $h\neq k$, $h,k\in[r]$ one has $i_h\neq i_k$ and $j_h\neq j_k$. 
K\"onig's Theorem relates the cardinality of an independent set of positions of a zero-one matrix to the minimum number of lines containing all the nonzero entries.

\begin{theorem}[K\"onig's Theorem,~\cite{Kon,Kon31}]\label{koenig}
If the entries of a rectangular matrix are zeros and ones, then the minimum number of lines containing all the entries equal to one is equal to the maximum cardinality of an independent set of positions corresponding to nonzero entries.
\end{theorem}

In \cite{Mes85}, Meshulam uses K\"onig's Theorem to establish a lower bound for the maximum rank of a matrix in a given vector space. In this section, we extend Meshulam's result from vector spaces of matrices to cosets. We start with a preliminary result.

\begin{lemma}\label{lemma2}
Let $D_1,\dots, D_r,A\in \F_q^{r\times r}$ such that for all $1\leq i\leq r$, the first $i-1$ rows of $D_i$ are zero and the $i$th row is the $i$th standard basis vector. Then 
there are $x_1,\dots,x_r\in\{0,1\}$ such that
\begin{equation*}
\rk\left(A+\sum_{i=1}^r x_iD_i\right)=r.
\end{equation*}
\end{lemma}

\begin{proof}
We proceed by induction on $r$. The case $r=1$ is trivial. 
Assume $r>1$. For $i\in[r-1]$ let $D_i'=D_i([r-1],[r-1])$. By the induction hypothesis, there exist $x_1,\dots,x_{r-1}\in\{0,1\}$ such that the matrix $A([r-1],[r-1])+\sum_{i=1}^{r-1}x_iD_i'$ is non-singular.
Since $D_r(i,j)=0$ for all $(i,j)\neq (r,r)$ and $D_r(r,r)=1$, by expanding with respect to the bottom row we obtain
\begin{equation*}
\begin{split}
&\det\left(A+\sum_{i=1}^{r-1}x_iD_i+D_r\right)=\det\left(A+\sum_{i=1}^{r-1}x_iD_i\right)
+{}\\&+(-1)^{r+1} \det\left(A([r-1],[r-1])+\sum_{i=1}^{r-1}x_iD_i'\right).
\end{split}		 		
\end{equation*}
The last summand is nonzero, therefore $$A+\sum_{i=1}^{r-1}x_iD_i+D_r\text{ and } A+\sum_{i=1}^{r-1}x_iD_i$$ cannot both be singular. 
\end{proof}
	
The next theorem extends the main result of~\cite{Mes85} from vector spaces to cosets.
	
\begin{theorem}\label{theoremMeshulam}
Let $A\in \F_q^{m\times n}$ and let $\mathcal{M}=\{M_1,\dots, M_d\}\subseteq\F_q^{m\times n}$. Then there exist $x_1,\ldots,x_d\in\{0,1\}$ such that 
\begin{equation*}
\rk(A+x_1M_1+\dots+x_dM_d)\geq\rho(\mathcal{M}).
\end{equation*}
\end{theorem}
	
\begin{proof}
Let $\rho(\mathcal{M})=r$. By Theorem~\ref{koenig} there exist $i_1,\dots, i_r\in[d]$ such that $\{\phi(M_{i_j}): j\in[r]\}$ is independent. Let $\phi(M_{i_j})=(s_j,l_j)$ for $j\in[r]$, then both $S=\{s_1,\dots,s_r\}$ and $L=\{l_1,\dots,l_r\}$ have cardinality $r$. 
		
We shall prove the theorem by showing that $A(S,L)+\langle B_1,\dots, B_r\rangle$ contains a non-singular matrix, where $B_j=M_{i_j}(S,L)$. We may assume that $s_1<s_2<\dots<s_r$. Let $\sigma$ be the permutation on $[r]$ for which $l_{\sigma(1)}<\dots<l_{\sigma(r)}$. Denote the $j$th row of $B_j$ by $b_j$. 
		
Clearly the first $j-1$ rows of $B_j$ are zero, $B_j(j,s)=0$ for $s\in[\sigma^{-1}(j)-1]$ and $B_j(j,\sigma^{-1}(j))\neq0$. Let $C\in\F_q^{r\times r}$ be the matrix with rows $b_1,\dots b_r$. Notice that $C$ is non-singular, since we can obtain an upper triangular matrix with nonzero entries on the diagonal by permuting the rows of $C$. Let $D_j=B_jC^{-1}$ for $j\in[r]$.
It easy to check that the first $j-1$ rows of $D_j$ are zero and the $j$th row is the $j$th standard basis vector, for all $j\in[r]$.

By Lemma \ref{lemma2} we have that $A(S,L)C^{-1}+\sum_{j=1}^{r}x_jD_j$ is non-singular for some $x_1,\dots, x_r\in\{0,1\}$. Therefore
\begin{equation*}
A(S,L)+\sum_{j=1}^{r}x_jB_j=\left(A(S,L)C^{-1}+\sum_{j=1}^{r}x_jD_j\right)C
\end{equation*}
is also non-singular. This implies that $\rk(A+\sum_{j=1}^{r}x_jM_{i_j})\geq r$.
\end{proof}

A theorem by Meshulam~\cite[Theorem~2]{Mes85} states that if $\mathcal{V}\subseteq\F_q^{m\times m}$ is an $\F_q$-linear subspace of $\dim(\mathcal{V})>mt$, then $\mathcal{V}$ contains a matrix of rank at least $t+1$. This result is easily generalized to $m\times n$ matrices. The next theorem extends Meshulam's results to cosets, i.e. sets of the form $A+\mathcal{V}$, where $\mathcal{V}\subseteq\F_q^{m\times n}$ is $\F_q$-linear and $A\in\F_q^{m\times n}$. The theorem was first shown by C. de Seguins Pazzis, see \cite[Corollary 2]{Pazzis2010TheAP}.

\begin{theorem}\label{affinespace}
Let $0\leq t<n$ and let $\mathcal{V}\subseteq \F_q^{m\times n}$ be an $\F_q$-linear subspace of $\dim(\mathcal{V})>mt$. Let $A\in \F_q^{m\times n}$. 
Then there exists $B\in \mathcal{V}$ such that
\begin{equation*}
\rk(A+B)\geq t+1.
\end{equation*}
Moreover, if $\{B_1,\dots, B_{mt+d}\}$ is a basis of $\mathcal{V}$, $d=\dim(\mathcal{V})-mt$, then $B$ can be chosen of the form $B=\sum_{i=1}^{mt+d}x_iB_i$ with $x_i\in\{0,1\}$.
\end{theorem}	

\begin{proof}
Let $\dim(\mathcal{V})=mt+d$ with $d>0$ and choose a basis $\{B_1,\dots, B_{mt+d}\}$ of $\mathcal{V}$. By performing Gaussian elimination on $\{B_1,\dots, B_{mt+d}\}$ we may assume that $\phi(B_1),\dots,\phi(B_{mt+d})$ are distinct. Since a line in a matrix cover at most $m$ entries we cannot cover $\phi(B_1),\dots,\phi(B_{mt+d})$ by less than $(mt+d)/m$ lines. Therefore, $$\rho(\{B_1,\dots, B_{mt+d}\})\geq t+1.$$ Theorem~\ref{theoremMeshulam} implies that there exists $B\in \mathcal{V}$ of the desired form, such that $\rk(A+B)\geq t+1$.
\end{proof}

Results on vector spaces are a special case of those on cosets. For example, the Anticode Bound is a direct consequence of 
Theorem~\ref{affinespace}.

\begin{theorem}[Anticode Bound, \cite{Rav16}]\label{rankmetric}
Let $\Cc\subseteq \F_q^{m\times n}$ be a rank-metric code. Then
\begin{equation*}
\dim(\Cc)\leq m\maxrk(\Cc).
\end{equation*}
\end{theorem}

If $A\in \mathcal{V}$ and $\mathcal{V}$ is a linear space, then $A+\mathcal{V}=\mathcal{V}$ and there exist linear spaces $\mathcal{V}\subseteq \F_q^{m\times n}$ such that $\dim(\mathcal{V})=mt$ and $\rk(A)\leq t$ for all $A\in \mathcal{V}$. Such linear spaces appear in the coding theory literature under the name of optimal anticodes. We now show that if $A\not\in \mathcal{V}$, that is if $A+\mathcal{V}\neq \mathcal{V}$, then every $\mathcal{V}$ of $\dim(\mathcal{V})=mt$ contains a $B$ such that $\rk(A+B)>t$. For odd $q$ this is an immediate consequence of Theorem~\ref{theoremMeshulam}, as we show in the next corollary. In Theorem~\ref{propq2} we prove the same result for any $q$. We choose to include Corollary~\ref{corollaryaffinespace}, since the proof is immediate.
	
\begin{corollary}\label{corollaryaffinespace}
Let $0\leq t<n$ and let $\mathcal{V}\subseteq \F_q^{m\times n}$ be an $\F_q$-linear subspace of dimension $\dim(\mathcal{V})= mt$. Let $A\in \F_q^{m\times n}\setminus \mathcal{V}$. If $q$ is odd, then there exists $B\in \mathcal{V}$ such that
\begin{equation*}
\rk(A+B)\geq t+1.
\end{equation*}
\end{corollary}

\begin{proof}
Let $\{B_1,\dots, B_{mt}\}$ be a basis of $\mathcal{V}$ and let $\bar{\mathcal{V}}=\langle A\rangle+\mathcal{V}$. Since $A\notin \mathcal{V}$, then $\dim(\bar{\mathcal{V}})=mt+1.$
By Theorem \ref{affinespace} there are $x_0,\dots,x_{mt}\in\{0,1\}$ such that
\begin{equation*}
\rk\left(A+x_0A+\sum_{i=1}^{mt}x_iB_i\right)\geq t+1.
\end{equation*}
Multiplying by $(1+x_0)^{-1}$ we find a matrix of the form $A+B$ with $A\not\in \mathcal{V}$, $B\in \mathcal{V}$ such that $\rk(A+B)\geq t+1$.
\end{proof}

The following lemma will be used in the proof of the next theorem.

\begin{lemma}\label{lemma3}
Let $f:\F_q^{r\times r}\to\F_q$ be a linear form that is constant on $\GL_r(\F_q)$. Suppose that either $r>1$ or $q\neq 2$. Then $f=0$.
\end{lemma}

\begin{proof}
Since $f$ is linear, there exist $a_{i,j}\in \F_q$, $i,j\in[r]$, such that
\begin{equation*}
f(X)=\sum_{1\leq i,j\leq r}a_{i,j}x_{i,j}
\end{equation*}
for any $X=(x_{i,j})\in \F_q^{r\times r}$.
If $r=1$ and $q\neq 2$, let $1\neq\alpha\in\F_q^*$. Then $f(\alpha)=f(1)-f(1-\alpha)=0$, hence $f=0$. 
If $r>1$, fix $(k,l)\in [r]\times[r]$. Let $B=(b_{i,j})$ be a permutation matrix such that $b_{k,l}=0$. Let $\bar B=B+E_{k,l}$.
Both $B$ and $\bar B$ are non-singular, so $f(B)=f(\bar B)$. Therefore $f(E_{k,l})=0$ by linearity. Since this is the case for every $(k,l)\in [r]\times[r]$, we conclude that $f=0$. 
\end{proof}

The next lemma will be used in the proof of Theorem \ref{propq2}.

\begin{lemma}\label{lemma:t=1}
	Let $n\geq 2$ and $m>2$. Let $\V\subseteq \F_2^{m\times n}$ be an $\F_2$-linear subspace such that $\dim(\V)= m$. Let $A\in \F_2^{m\times n}\setminus \V$. Then there exists $B\in \V$ such that
	\begin{equation*}
	\rk(A+B)\geq 2.
	\end{equation*}
\end{lemma}

\begin{proof}
For $v\in\F_q^m$ and $w\in\F_q^m$, denote by $v\otimes w$ the $m\times n$ matrix whose entry in position $(i,j)$ is $v_iw_j$. 
If $\rk(A)\geq2$, then the statement holds with $B = 0$. If $\rk(A) = 1$, then up to equivalence we may assume that $A=E_{1,1}=e_1\otimes e_1$. 
If $\maxrk(\V)=1$, then $\V$ is an optimal anticode and the statement holds. If $\maxrk(\mathcal{V})>2$, then there exists $B\in \V$ with $\rk(B)>2$. Hence $\rk(A+B)\geq 2$, since $A$ has rank 1. Therefore, it suffices to prove the statement for $\maxrk(\V)=2$.
	
First suppose that there are two different elements $V_1,V_2\in\V$ of rank 1. Write $V_1=v_1 \otimes w_1$ and $V_2=v_2\otimes w_2$ for some $v_1, v_2 \in \F_2^m$ and $ w_1, w_2 \in \F_2^n$. If $\rk(A+V_1)=\rk(A+V_2)=1$, then either $v_1=e_1$ or $w_1=e_1$ and either $v_2=e_1$ or $w_2=e_1$. If either $v_1=e_1$ and $w_2=e_1$, or $w_1=e_1$ and $v_2=e_1$, then $\rk(A+V_1+V_2)=2$, since $V_1, V_2 \neq A$. 
If instead $v_1=v_2=e_1$, then $e_1,w_1,w_2$ are linearly independent and every matrix in $\langle A,A+V_1,A+V_2 \rangle$ has rank 1.
Let $B\in\V$ be an element of rank two. Then one of the vectors $e_1,e_1+w_1,e_1+w_2\not\in\mathrm{rowsp}(B)$. Therefore, there exists $C\in\{A,A+V_1,A+V_2\}$ such that $\rk(C+B)=\dim(\mathrm{rowsp}(C+B))\geq2$. In the case where $w_1=w_2=e_1$, we proceed similarly using the column space.
	
Suppose now that in $\V$ there is at most one element of rank 1. Then every linear combination with an element of maximum rank in $\V$ has again maximum rank. Hence, since $\dim(\V) >2$, there are two linearly independent elements $B_1,B_2$ such that $\rk(B_1)=\rk(B_2)=\rk(B_1+B_2)=2$.
If $\rk(A+B_1)=\rk(A+B_2)=\rk(A+B_1+B_2)=1$, then
$$B_1=e_1\otimes e_1+e_2\otimes e_2, \ B_2=e_1\otimes e_1+v_2\otimes w_2, \ B_1+B_2=e_1\otimes e_1+v_3\otimes w_3,$$ possibly after applying a code equivalence that fixes $A$. Since
$$B_2=e_1\otimes e_1+v_2\otimes w_2=e_2\otimes e_2+v_3\otimes w_3$$ and 
$$B_1+B_2=e_1\otimes e_1+v_3\otimes w_3=e_2\otimes e_2+v_2\otimes w_2$$
have rank $2$, then $v_2,v_3\not\in\{e_1,e_2\}$.
Moreover
$$e_1\otimes e_1+e_2\otimes e_2= v_2\otimes w_2+v_3\otimes w_3,$$
hence $\langle v_2,v_3\rangle=\langle e_1,e_2\rangle$. The only possibility is that $v_2=v_3=e_1+e_2$, but this contradicts the assumption that $B_1 = v_2\otimes w_2+v_3\otimes w_3$ has rank 2. Therefore, one among $A+B_1,A+B_2,A+B_1+B_2$ has rank at least 2.
\end{proof}

The next example we show that the condition $m>2$ in Lemma \ref{lemma:t=1} is necessary. The example is essentially the same as the example that appears below Theorem 2 in \cite{Pazzis2010TheAP}. 

\begin{example}
Consider the $2$-dimensional space $\V\subseteq\F_2^{2\times 2}$ given by
$$\V=\left\langle\begin{pmatrix}
	1&0\\0&1
	\end{pmatrix},
	\begin{pmatrix}
	0&0\\1&0
	\end{pmatrix}\right\rangle$$
	and let $A=\begin{pmatrix}
	1&0\\0&0
	\end{pmatrix} \notin \V$. Then $\max_{B\in\V}\{ \rk(A+B) \}=1$.
\end{example}

The next theorem generalizes Corollary~\ref{corollaryaffinespace} to any $q$. It was first shown by C. de Seguins Pazzis, see \cite[Corollary 2]{Pazzis2010TheAP}.

\begin{theorem}\label{propq2}
Let $0\leq t<n$ and let $\V\subseteq \F_q^{m\times n}$ be an $\F_q$-linear subspace such that $\dim(\V)= mt$. Let $A\in \F_q^{m\times n}\setminus \V$. If either $t\neq1$ or $m\neq2$ or $q\neq 2$ or $\rk(A)\neq1$, then there exists $B\in \V$ such that
\begin{equation*}
\rk(A+B)\geq t+1.
\end{equation*}
\end{theorem}

\begin{proof}
If $t=0$, then $\V=0$ and the thesis is readily verified. 
Suppose that $t\geq 1$ and let $\mathcal{M}=\{M_1,\dots,M_{mt}\}$ be a basis of $\V$. 
Up to a change of basis, we may assume without loss of generality that $\phi(M_i)\neq\phi(M_j)$ if $i\neq j$. In particular, $\rho(\mathcal{M})\geq t$.
If $\rho(\mathcal{M})\geq t+1$, then we conclude by Theorem \ref{theoremMeshulam}. 

Suppose that $\rho(\mathcal{M})=t$. Up to code equivalence, we may assume that the $t$ lines that cover $\phi(M_i)$ for all $i$ are the first $t$ columns. If $t=1$ and $\rk(A)\geq 2$, then let $B=0$. If $t=1$, $\rk(A)=1$, $q\neq2$ and there exists $B\in\V$ with $\rk(B)\geq2$, then either $\rk(A+B)\geq2$ or $\rk(A+2B)\geq2$. If $t=1$, $\rk(A)=1$, $q\neq2$, and $\maxrk\V=1$, then $\V$ is an optimal anticode and the result follows easily.
If $q=2$ we conclude by Lemma \ref{lemma:t=1}, since $m\neq 2$.

Suppose now that $\rho(\mathcal{M})=t\geq 2$.
For every $t+1\leq l\leq n$ and every $k\in[m]$ there exists a linear form $f_{k,l}\in\F_q[x_{i,j}\mid (i,j)\in[m]\times[t]]$ 
such that
\begin{equation*}
\V=\{(x_{u,v})_{u,v}\in\F_q^{m\times n}: x_{k,l}=f_{k,l}(x_{i,j}) \text{ for all } k\in[m], l\in[n]\setminus[t]\}.
\end{equation*}
Assume without loss of generality that the entry of $M_i$ in position $\phi(M_i)$ is 1. Then $f_{k,l}$ is obtained by writing a matrix of $\V$ as $\sum_{(i,j)\in[m]\times[t]}x_{i,j}M_{(i-1)t+j}$.
Assume that $\max_{B\in\mathcal{V}}\rk(A+B)=t$ for some $A\in\F_q^{m\times n}$. It suffices to show that $A\in \V$. Up to reducing $A$ modulo $\V$, we may assume without loss of generality that $a_{i,j}=0$ for $(i,j)\in[m]\times[t]$.
Fix $(k,l)\in[m]\times[n]$ with $l\geq t+1$.
Let $X=(x_{i,j})_{i,j}\in \V$. We have that 
\begin{equation*}
x_{k,l}+a_{k,l}=f_{k,l}(x_{i,j})+a_{k,l}.
\end{equation*}
Let $L=[t]$ and let $S$ be a subset of $[m]\setminus\{k\}$ of cardinality $t$.
Let $x_{i,j}=0$ for $i\notin S$ and $j\in L$. For any choice of $(x_{i,j})_{i\in S, j\in L}$ such that $X(S,L)+A(S,L)=X(S,L)$ is invertible, one has 
\begin{equation}\label{eqn:ahk}
0=x_{k,l}+a_{k,l}=f_{k,l}(x_{i,j})+a_{k,l},
\end{equation} 
since every matrix in $(A+\V)(S\cup\{k\},L\cup\{l\})$ has rank smaller than or equal to $t$.
Lemma \ref{lemma3} together with (\ref{eqn:ahk}) implies that $a_{k,l}=0$.
This proves that $A\in \V$.
\end{proof}

\section{Anticode Bound and optimal anticodes}\label{sec:OAC} 

In this section we prove an Anticode Bound for sum-rank metric codes. Our bound improves the bound from~\cite[Theorem~2.2]{BGROAC}. 

\begin{theorem}[Anticode Bound]\label{theoremab}
Let $\Cc\subseteq \MM$ be an $\F_q$-linear subspace. Then 
\begin{equation}\label{anticodebound}
\dim(\Cc)\leq \max_{C\in\Cc} \left\{\sum_{i=1}^\ell m_i \rk(C_i)\right\}.
\end{equation}
In particular, if $m_1=\ldots=m_\ell=m$, then $$\dim(\Cc)\leq m\maxsrk(\Cc).$$
\end{theorem}

\begin{proof}
We proceed by induction on $\ell$. If $\ell=1$, then $\Cc$ is a rank-metric code, the sum-rank metric coincides with the rank metric, and the statement is Theorem~\ref{rankmetric}.

Let $\ell>1$. Let $\pi$ be the canonical projection from $\MM$ onto $\F_q^{m_1\times n_1}\times\dots\times \F_q^{m_{\ell-1}\times n_{\ell-1}}$ 
and let $\pi_{\ell}$ be the canonical projection from $\MM$ onto $\F_q^{m_\ell\times n_\ell}$. 
Define $\mathcal{A}= \pi(\Cc)$ and $\mathcal{B}= \pi_{\ell}(\pi^{-1}(0)\cap\Cc)$ and let $\tilde \Cc= \mathcal{A}\times \mathcal{B}$. Since $\dim(\pi^{-1}(0)\cap\Cc)=\dim(\pi_\ell(\pi^{-1}(0)\cap\Cc))=\dim(\mathcal{B})$, we have that
\begin{equation*}
\dim(\Cc)=\dim(\mathcal{A})+\dim(\pi^{-1}(0)\cap\Cc)=\dim(\tilde \Cc).
\end{equation*}
By the induction hypothesis there is $(C_1,\ldots,C_{\ell-1})\in \mathcal{A}$ such that 
\begin{equation*}
\sum_{i=1}^{\ell-1} m_i \rk(C_i)\geq \dim(\A)=\dim(\Cc)-\dim(\mathcal{B}).
\end{equation*}
Let $C_{\ell}\in \pi_{\ell}(\Cc)$ such that $(C_1,\dots,C_{\ell})\in \Cc$. By Theorem \ref{affinespace} there is a $B\in \mathcal{B}$ such that
$$\rk(C_{\ell}+ B)\geq \left\lceil\frac{\dim(\mathcal{B})}{m_\ell}\right\rceil.$$
Therefore
\begin{equation*}
\begin{split}
\sum_{i=1}^{\ell-1}m_i\rk(C_i)+m_\ell\rk(C_\ell+{B}) & \geq \dim(\Cc)-\dim(\mathcal{B})
+m_\ell\left\lceil\frac{\dim(\mathcal{B})}{m_\ell}\right\rceil
\\
& \geq\dim(\Cc).
\end{split}
\end{equation*}
The element $(C_1,\ldots,C_{\ell-1},C_\ell+B)\in\Cc$, since $(C_1,\ldots,C_{\ell-1},C_\ell)\in\Cc$ and $B\in\mathcal{B}$. This concludes the proof.
\end{proof}

Optimal sum-rank metric anticodes may now be defined as the codes which meet the Anticode Bound. 

\begin{definition}\label{defn:OAC}
A sum-rank metric code $\Cc\subseteq\MM$ is an optimal anticode if 
\begin{equation*}
\dim(\Cc)=\max_{C\in\Cc}\left\{ \sum_{i=1}^\ell m_i \rk(C_i)\right\}.
\end{equation*}
\end{definition}

\begin{remark}
In \cite{BGROAC}, the authors give a definition of $r$-anticode for $r$ a non-negative integer. In~\cite[Theorem~2.2]{BGROAC} they establish an upper bound for the dimension of an $r$-anticode. For a given $\Cc\subseteq \MM$ and $r=\maxsrk(\Cc)$, \cite[Theorem~2.2]{BGROAC} yields
\begin{equation}\label{ABBGR}
\dim(\Cc)\leq \max\left\{ \sum_{i=1}^\ell m_iu_i : \sum_{i=1}^\ell u_i=\maxsrk(\Cc), u_i\leq n_i \mbox{ for all } i\right\}.
\end{equation}
Notice that our Anticode Bound is tighter than (\ref{ABBGR}), since for all $C=(C_1,\ldots,C_\ell)\in\Cc$ there exist $u_1,\ldots,u_\ell\in\mathbb{Z}$ such that $\sum_{i=1}^\ell u_i=\maxsrk(\Cc)$ and $\rk(C_i)\leq u_i\leq n_i$ for all $i$. In particular, all codes that meet bound (\ref{ABBGR}) also meet our Anticode Bound. Moreover the bounds are different, as one can easily check by comparing Theorem \ref{thm:OAC} in this paper and \cite[Corollary 3.8]{BGROAC}. In \cite[Definition 2.3]{BGROAC}, the authors define optimal anticodes as those that meet the bound (\ref{ABBGR}). In particular, an optimal anticode according to \cite{BGROAC} is an optimal anticode according to Definition \ref{defn:OAC}, but the converse is not true in general. For example, the code $0\times\F_2\subseteq\F_2^{2\times 2}\times\F_2$ is an optimal anticode according to Definition \ref{defn:OAC}, but it does not meet (\ref{ABBGR}).
\end{remark}

A simple computation allows one to show that if $\Cc_i\subseteq\F_q^{m_i\times n_i}$ is an optimal anticode with respect to the rank metric for $i\in[\ell]$, then $\Cc_1\times\dots\times\Cc_\ell\subseteq\MM$ is an optimal anticode with respect to the sum-rank metric. Moreover, one has the following.
	
\begin{proposition}\label{propoac}
Let $\Cc\subseteq\MM$ be an optimal anticode and assume that $m_1=\ldots=m_\ell=m$. For $i\in[\ell]$ let $\pi_i:\MM\rightarrow\F_q^{m_i\times n_i}$ be the canonical projection. The following are equivalent:
\begin{enumerate}[$(1)$]
\item $\Cc=\Cc_1\times\dots\times\Cc_\ell$ and $\Cc_i$ is an optimal rank-metric anticode for $i\in[\ell]$.
\item $\maxsrk(\Cc)=\sum_{i=1}^{\ell}\maxrk(\pi_i(\Cc))$.
\end{enumerate}
\end{proposition}
	
\begin{proof}
$(1)\implies (2)$ follows from a simple computation.

\noindent $(2)\implies (1)$ Clearly, $\Cc\subseteq \prod_{i=1}^\ell\pi_i(\Cc)$, so
\begin{equation*}
m\maxsrk(\Cc)=\dim(\Cc)\leq\sum_{i=1}^\ell\dim(\pi_i(\Cc))\leq\sum_{i=1}^{\ell}m\maxrk(\pi_i(\Cc)).
\end{equation*}
Since $\maxsrk(\Cc)=\sum_{i=1}^{\ell}\maxrk(\pi_i(\Cc))$, we have that
\begin{equation*}
\dim(\Cc)= \dim\left(\prod_{i=1}^\ell\pi_i(\Cc)\right)\text{ and } \dim(\pi_i(C))=m\maxrk(\pi_i(\Cc)).
\end{equation*}
Therefore $\Cc= \prod_{i=1}^\ell\pi_i(\Cc)$ and $\Cc_i$ is an optimal rank-metric anticode for all $i\in[\ell]$.
\end{proof}

We will prove that optimal anticodes in the sum-rank metric are generated by their elements of maximum sum rank. We start by proving the result in the special case of rank-metric anticodes.

\begin{lemma}\label{lemmageneratingset}
Let $\Cc\subseteq\F_q^{m\times n}$ be an optimal anticode. Then $\Cc$ is generated by its elements of maximum rank.
\end{lemma}

\begin{proof}
Let $t=\maxrk(\Cc)$, then $\dim(\Cc)=mt$.
Up to code equivalence we may assume that $\Cc$ consists of all matrices whose rowspace is contained in $\langle e_1,\ldots,e_t\rangle$, where $e_1,\ldots,e_t\in\F_q^n$ are the first $t$ elements of the standard basis.
Therefore it suffices to prove the statement for $\Cc=\F_q^{m\times t}$. 
Let $\{E_{i,j}\}_{1\leq i\leq m,1\leq j\leq t}$ be the standard basis of $\F_q^{m\times t}$. Let $I=\sum_{i=1}^t E_{i,i}\in\F_q^{m\times t}$. For each $(i,j)\in[m]\times[t]$ there exists a permutation matrix $S_{i,j}\in\F_q^{m\times m}$ such that $(S_{i,j}I)_{i,j}=0$. 
Therefore one can write $E_{i,j}=(S_{i,j}I+E_{i,j})-S_{i,j}I$, with 
$\rk(S_{i,j}I)=\rk(S_{i,j}I+E_{i,j})=t$. This implies that  $\{S_{i,j}I+E_{i,j},S_{i,j}I\}_{1\leq i\leq m,1\leq j\leq t}$ is a set of matrices of rank $t$ which generates $\F_q^{m\times t}$. 
\end{proof}

The next observations will be useful in order to extend the result of Lemma \ref{lemmageneratingset} to optimal anticodes in the sum-rank metric.

\begin{lemma}\label{maxrank_gen}
Let $m\geq 2$ and let $\Cc\subseteq\F_2^{m\times n}$ be an optimal rank-metric anticode of $\maxrk(\Cc)=t$. Then every element of $\Cc$ of rank $t$ can be written as the sum of two elements of $\Cc$ of rank $t$.
\end{lemma}

\begin{proof}
Up to code equivalence we may assume that $\Cc$ consists of all matrices whose rowspace is contained in $\langle e_1,\ldots,e_t\rangle$, where $e_1,\ldots,e_t\in\F_2^n$ are the first $t$ elements of the standard basis. Therefore, it suffices to show that every element of full rank in $\F_2^{m\times t}$ can be written as the sum of two elements of $\F_2^{m\times t}$ of full rank. Let $C=(c_1,\ldots,c_t)\in\F_2^{m\times t}$ be the matrix whose columns are $c_1,\ldots,c_t\in\F_2^m$. Assume that $\rk(C)=t$. If $t=1$, let $\tilde{C}\in\Cc\setminus\{C,0\}$. Notice that $\tilde{C}$ exists, since $m\geq 2$. Then $\tilde{C},C+\tilde{C}$ are elements of rank 1 and $C=\tilde{C}+(C+\tilde{C})$. 
If $t$ is even, then $C=C_1+C_2$ where \begin{eqnarray*}
C_1=(c_1+c_2,c_1,c_3+c_4,c_3,\ldots,c_{t-1}+c_t,c_{t-1}),\\ C_2=(c_2,c_1+c_2,c_4,c_3+c_4,\ldots,c_t,c_{t-1}+c_t).\end{eqnarray*}
If $t\neq 1$ is odd, then $C=C_1+C_2$ where \begin{eqnarray*}
C_1=(c_1+c_2,c_3,c_1,c_4+c_5,c_4,\ldots,c_{t-1}+c_t,c_{t-1}),\\ C_2=(c_2,c_3+c_2,c_1+c_3,c_5,c_4+c_5,\ldots,c_t,c_{t-1}+c_t).
\end{eqnarray*}
Since $C_1$ and $C_2$ have the same column space as $C$, they have full rank.
\end{proof}

\begin{theorem}\label{corollarygeneratingset}
Let $\Cc=\Cc_1\times\dots\times\Cc_\ell\subseteq\MM$, where $\Cc_i$ is an optimal rank-metric anticode for all $i\in[\ell]$. If either $m_{\ell-1}\geq 2$ or $q\neq 2$, then $\Cc$ is generated by its elements of maximum sum-rank. 
\end{theorem}

\begin{proof}
Let $C=(C_1,\ldots,C_\ell)\in\Cc$ be such that $$\sum_{i=1}^\ell m_i\rk(C_i)=\max_{D\in\Cc}\left\{ \sum_{i=1}^\ell m_i \rk(D_i)\right\}.$$ 
Since $\Cc$ is a product, then $C_i$ is an element of maximum rank in $\Cc_i$ for all $1\leq i\leq \ell$. If $q\neq 2$, let $\alpha\in\F_q\setminus\{0,1\}$. Then $$(0,\ldots,0,C_i,0,\ldots,0)\in\langle(C_1,\ldots,C_\ell),(C_1,\ldots,C_{i-1},\alpha C_i,C_{i+1},\ldots,C_\ell)\rangle.$$ 
Therefore $\Cc$ is generated by its element of maximum sum-rank, since each $\Cc_i$ is generated by its elements of maximum rank by Lemma \ref{lemmageneratingset}. 

If $q=2$ and $i\neq \ell$, then by Lemma~\ref{maxrank_gen} there exist $C_i^\prime,C_i^{\prime\prime}\in\Cc_i$ of maximum rank such that $C_i=C_i^\prime+C_i^{\prime\prime}$. Let $C^\prime=(C_1,\ldots,C_{i-1},C_i^\prime,C_{i+1},\ldots,C_\ell)$ and $C^{\prime\prime}=(C_1,\ldots,C_{i-1},C_i^{\prime\prime},C_{i+1},\ldots,C_\ell)$. Then $$(0,\ldots,0,C_i,0,\ldots,0)\in\langle C^\prime, C^{\prime\prime}\rangle.$$ 
Since $C$ and $(0,\ldots,0,C_i,0,\ldots,0)$, $i\in[\ell-1]$, belong to the subcode of $\Cc$ generated by its codewords of maximum sum-rank, then also $(0,\ldots,0,C_\ell)$ does. Therefore $\Cc$ is generated by its element of maximum sum-rank.
\end{proof}

\begin{example}
For $\ell\geq 2$ and $m_{\ell-1}=1$, the code $\Cc=0\oplus\ldots\oplus 0 \oplus\F_2\oplus\F_2$ is an optimal anticode, which is not generated by its unique element $(0,\ldots,0,1,1)$ of maximum sum-rank.
\end{example}

The next result on generating sets of optimal binary anticodes in the Hamming metric will also be useful. 

\begin{lemma}\label{genshighwt}
Let $\Cc\subseteq\F_2^\ell$ be an optimal anticode of $\dim(\Cc)=t\geq 1$. 
Then $\Cc$ is generated by its elements of weight $t$ and $t-1$.
\end{lemma}

\begin{proof}

Let $G$ be a generator matrix of $\Cc$ and assume that $G$ is in reduced row echelon form. Denote by $g_1,\ldots,g_t$ the rows of $G$. Let $v=g_1+\ldots+g_t$. Then the vectors $v,v+g_1,\ldots,v+g_t$ have weight $t-1$ or $t$ and are a system of generators of $\Cc$, since $g_i=v+(v+g_i)$ for all $i$.
\end{proof}

The following technical lemma will be used in the proof of Theorem \ref{thm:OAC}.

\begin{lemma}\label{lemma:case2}
Let $q=2$, $\ell\geq 2$, $m_1 = n_1 = 2$, and let $k=\max\{i\in[\ell]\mid m_i>1\}$. Let $\Cc \subseteq \MM,\A = \pi(\Cc), \mathcal{B} = \pi_1(\pi^{-1}(0) \cap \Cc)$, where $\pi:\MM\rightarrow\F_2^{m_2\times n_2}\times\dots\times \F_2^{m_{\ell}\times n_{\ell}}$ and $\pi_1:\MM\rightarrow\F_2^{m_1\times n_1}$ are the canonical projections.
If $\dim(\mathcal{B}) = 2$ and $\A = \prod_{i=2}^k \Cc_i \times \Cc'$  for optimal anticodes $\Cc' \subseteq \F_2^{\ell - k}$ and $\Cc_i \subseteq \F_2^{m_i \times n_i}$ for all $i \in [k] \setminus \{ 1\}$, then one of the following holds:
\begin{itemize}
\item[(i)] $\mathcal{B}$ is an optimal anticode,
\item[(ii)] There is $B \in \mathcal{B}$ and $C=(C_1,\ldots,C_\ell)\in \Cc$ with $\sum_{i=2}^k m_i\rk(C_i)+\mathrm{wt}(C_{k+1},\dots,C_{\ell})\geq\sum_{i=2}^{k} m_i\maxrk(\Cc_i)+\maxwt(\Cc')-1$, such that
$$\rk(B+C_1)=2.$$
\end{itemize}
\end{lemma}

\begin{proof}
If $\maxrk(\mathcal{B}) = 1$, then $\mathcal{B}$ is an optimal anticode. Assume therefore that $\maxrk(\mathcal{B}) = 2$. 
Let $G$ be a generator matrix of $\Cc'$ and assume that $G$ is in reduced row echelon form and that $\dim(\Cc')=t$. If $t=0$, then 
let $D=(D_2, \dots , D_k), E=(E_2, \dots , E_k) \in \prod_{i=2}^k \Cc_i$ be codewords such that each component of $D,E,D+E$ has maximal rank. Such matrices exist, since each $\Cc_i$ is an optimal anticode. Let $D_1, E_1 \in \Cc_1$ be such that $(D_1, \dots , D_k,0)$, $(E_1, \dots , E_k,0) \in \Cc$. If one of $D_1, E_1 $, and $D_1 + E_1$ is zero, then we conclude by taking $B$ of rank 2.  If there is a rank 2 element among $D_1,E_1$, and $D_1 + E_1$, then we conclude by taking $B = 0$. If $D_1, E_1, D_1 + E_1$ all have rank 1, then again we easily conclude. In fact, either $\langle D_1,E_1\rangle\cap\mathcal{B}\neq 0$, or $$|(D_1+\mathcal{B})\cup(E_1+\mathcal{B})\cup(D_1+E_1+\mathcal{B})|=12,$$ but in $\F_2^{2\times 2}$ we have only 9 elements of rank 1. 

Suppose now that $t\geq 1$ and let $g_1,\dots,g_t,v$ as in the proof of Lemma \ref{genshighwt}. For every $i\in[t]$, there exists $G^i_1\in \F^{2\times2}_2$ such that $G^i=(G^i_1,0,\dots,0,g_i)\in\Cc$. If for every $i\in[t]$, $G^i_1\in\mathcal{B}$, then $\Cc=\mathcal{D}\times\Cc'$, where $\mathcal{D}\subseteq\F_2^{m_{1}\times n_{1}}\times\ldots\times\F_2^{m_{k}\times n_{k}}$. Therefore we reduce to the situation $t=0$, which we treated above.
Hence we assume without loss of generality that $G^1_1\notin\mathcal{B}$. Let $C=\sum_{i=1}^tG^i=(C_1,0,\dots,0,v)$ and $D=(D_1,D_2,\dots,D_k,0,\dots,0)$ such that $D_i$ has max rank in $\Cc_i$ for $i\in[k]\setminus\{1\}$. If either $D_1+C_1$ or $D_1+C_1+G_1^1$ belongs to $\mathcal{B}$, then we conclude.
If $D_1+C_1$,$D_1+C_1+G_1^1\notin\mathcal{B}$, then since $G_1^1\notin\mathcal{B}$, we have that
$$(D_1+C_1+G_1^1+\mathcal{B})\cap(D_1+C_1+\mathcal{B})=\emptyset,$$
and 
$$((D_1+C_1+G_1^1+\mathcal{B})\cup(D_1+C_1+\mathcal{B}))\cap(\langle G_1^1\rangle+\mathcal{B})=\emptyset.$$
Notice that $\F_2^{2\times 2}$ consists of the zero matrix, 9 elements of rank 1, and 6 elements of rank 2. 
In $\langle G_1^1\rangle+\mathcal{B}$  there are at least two elements of rank $1$, since $\dim(\langle G_1^1\rangle+\mathcal{B})=3$. Therefore, in $(D_1+C_1+G_1^1+\mathcal{B})\cup(D_1+C_1+\mathcal{B})$ there must be at least an element of rank 2. We conclude, since the elements $D+C$ and $D+C+G^1$ satisfy the condition from (ii).
\end{proof}

In the next theorem we show that the optimal anticodes in the sum-rank metric are products of optimal anticodes in the rank metric and an optimal anticode in the Hamming metric.

\begin{theorem}\label{thm:OAC}
Let $k=0$ if $m_1=1$ and $k=\max\{i\in[\ell]\mid m_i>1\}$ otherwise. 
A code $\Cc\subseteq\MM$ is an optimal anticode if and only if there is an optimal anticode $\Cc^\prime\subseteq\F_q^{\ell-k}$ and optimal anticodes $\Cc_i\subseteq\F_q^{m_i\times n_i}$ for all $i\in[k]$ such that $\Cc=\prod_{i=1}^k\Cc_i\times\Cc^\prime$. 
\end{theorem}

\begin{proof}
Assume that $\Cc^\prime\subseteq\F_q^{\ell-k}$ is an optimal Hamming-metric anticode and $\Cc_i\subseteq\F_q^{m_i\times n_i}$ are optimal rank-metric anticodes for $i\in[k]$. It is straightforward to prove that $\Cc=\prod_{i=1}^k\Cc_i\times\Cc^\prime\subseteq \MM$ is an optimal anticode. Further, the statement of the theorem holds for $m_1=\ldots=m_\ell=1$. Therefore, we may assume that $m_1>1$, hence also $k\geq 1$. We proceed by induction on $\ell$. For $\ell=1$, the theorem holds trivially.

We suppose that the theorem holds for $\ell-1$ and we prove it for $\ell>1$. Let $\mathcal{A}=\pi(\Cc)$, $\mathcal{B}=\pi_{1}(\pi^{-1}(0)\cap\Cc)$, and $\tilde \Cc= \mathcal{B}\times\mathcal{A}$.
As in the proof of Theorem \ref{theoremab}, we have 
\begin{equation}\label{eqn:A&B}
\dim(\Cc)=\dim(\mathcal{B})+\dim(\A)\leq m_1 \rk(C_1+B)+\sum_{i=2}^{\ell} m_i\rk(C_i),
\end{equation} 
where $(C_1,\ldots,C_\ell)\in\Cc$, $(C_2,\ldots,C_\ell)$ maximizes $\sum_{i=2}^{\ell}m_i\rk(C_i)$ on $\mathcal{A}$, and $m_1 \rk(C_1)\geq\dim(\mathcal{B})$. Since $(C_1,C_2,\ldots,C_\ell)\in\Cc$ and $\Cc$ is an optimal anticode, then (\ref{eqn:A&B}) is an equality. In particular, 
\begin{equation}\label{eq:dims}
m_1 \rk(C_1)=\dim(\mathcal{B}) \;\;\mbox{ and }\;\; \sum_{i=2}^{\ell}m_i\rk(C_i)=\dim(\A).
\end{equation} 
This proves that $\A$ is an optimal anticode. Therefore, by the induction hypothesis, there is an optimal anticode $\Cc^\prime\subseteq\F_q^{\ell-k}$ and optimal anticodes $\Cc_i\subseteq\F_q^{m_i\times n_i}$ for $2\leq i\leq k$ such that $\A=\prod_{i=2}^k\Cc_i\times\Cc^\prime$.

We claim that $C_1\in\mathcal{B}$. In fact, if $C_1\not\in\mathcal{B}$ and either $\dim(\mathcal{B})\neq2$ or $m\neq2$ or $q\neq 2$ or $n\neq2$ or $\rk(C_1)\neq1$, then by Theorem \ref{propq2} there exists $B\in\mathcal{B}$ such that $$\rk(C_1+B)>\dim(\mathcal{B})/m_1=\rk(C_1).$$ Since $B\in\mathcal{B}$, then $(C_1+B,C_2,\ldots,C_\ell)\in\Cc$. However, this contradicts the optimality of $\Cc$, since $m_1\rk(C_1+B)+\sum_{i=2}^{\ell}m_i\rk(C_i)>\sum_{i=1}^{\ell}m_i\rk(C_i)=\dim(\Cc)$. This proves that $C_1\in\mathcal{B}$, so $C_1+B\in\mathcal{B}$, hence $\mathcal{B}$ is an optimal anticode by (\ref{eq:dims}).
If $\dim(\mathcal{B})=2$, $m=n=2$, $q= 2$, and $\rk(C_1)=1$, then by Lemma \ref{lemma:case2} either $\mathcal{B}$ is an optimal anticode, or there exists $B\in\mathcal{B}$ and $\bar C=(\bar C_1,\dots,\bar C_{\ell})$ such that $(\bar C_1+B,\dots,\bar C_{\ell})\in\Cc$ and
$$m_1\rk(\bar C_1+B)+\sum_{i=2}^{\ell}m_i\rk(\bar C_i)\geq 4+\dim(\mathcal{A})-1=\dim(\Cc)+1.$$
This is a contradiction, since $\Cc$ is an optimal anticode. We conclude that also in this case $\mathcal{B}$ is an optimal anticode.
In addition, our arguments show that, if $(C_1,\ldots,C_\ell)\in\Cc$ is such that $(C_2,\ldots,C_\ell)$ maximizes $\sum_{i=2}^{\ell}m_i\rk(C_i)$, then $C_1\in\mathcal{B}$.  Hence $(0,C_2,\ldots,C_\ell)\in\Cc$.

In order to conclude the proof, it suffices to show that $\Cc=\mathcal{B}\times\A$. Since $\Cc\supseteq\mathcal{B}\times 0$, it suffices to show that $\Cc\supseteq 0\times\A$. If either $k\geq\ell-1$ or $q\neq 2$, then $0\times \A$ is generated by its element of maximum sum-rank by Theorem~\ref{corollarygeneratingset}. Since these belong to $\Cc$, we have that $0\times \A\subseteq\Cc$.
Therefore, assume that $k\leq\ell-2$ and $q=2$.
Let $2\leq i\leq k$. By Lemma \ref{maxrank_gen}, if $C_i\in\Cc_i$ is an element of maximum rank, then $C_i=D_i+D_i^\prime$ for some $D_i,D_i^\prime\in\Cc_i$ of maximum rank. 
Hence $$(0,\ldots,0,C_i,0,\ldots,0)=(0,D_2,\ldots,D_k,D)+(0,D_2^\prime,\ldots,D_k^\prime,D)$$ where $D_j=D_j^\prime\in\Cc_j$ is an element of maximum rank for any $j\in\{2,\ldots,k\}\setminus\{i\}$ and $D$ is an element of maximum rank of $\Cc^\prime$. Since $(0,D_2,\ldots,D_k,D)$, $(0,D_2^\prime,\ldots,D_k^\prime,D)$ are elements of maximum sum-rank in $0\times\A$, they belong to $\Cc$.
This proves that, for any $2\leq i\leq k$, if $C_i$ has maximum rank among the elements of $\Cc_i$, then \begin{equation}\label{eqn:Ci}
(0,\ldots,0,C_i,0,\ldots,0)\in\Cc.
\end{equation} 
Since $\Cc_i$ is generated by its elements of maximum rank by Lemma \ref{lemmageneratingset}, then $$0\times\ldots\times 0\times\Cc_i\times 0\times\ldots\times 0\subseteq\Cc$$ for all $2\leq i\leq k$.

In addition, it follows from (\ref{eqn:Ci}) that $(0,\ldots,0,D)\in\Cc$ for any $D\in\Cc^\prime$ of maximum Hamming weight.
We claim that $0\times\ldots\times 0\times\Cc^\prime\subseteq\Cc$. 
Let $t$ be the maximum weight of a codeword in $\Cc^\prime$ and let $D'\in\Cc^\prime$ be an element of weight $t-1$. By Lemma \ref{genshighwt} it suffices to show that $(0,\ldots,0,D')\in\Cc$. 
Let $(D_2,\ldots,D_k,D')\in\A$ with $\rk(D_i)=\maxrk(\Cc_i)$ for $2\leq i\leq k$. Let $D_1$ be such that $(D_1,D_2,\ldots,D_k,D')\in\Cc$. Since $0\times\Cc_2\times\ldots\times\Cc_k\times 0\subseteq\Cc$, then $(D_1,0,\ldots,0,D')\in\Cc$. If $D_1\in\mathcal{B}$ the claim follows, since $(D_1,0,\ldots,0)\in\mathcal{B}\times 0\subseteq\Cc$. If $D_1\not\in\mathcal{B}$, then since $\mathcal{B}$ is an optimal anticode, there exists $B\in\mathcal{B}$ such that $\rk(B+D_1)\geq\maxrk(\mathcal{B})+1$. Then the element $(B+D_1,D_2,\ldots,D_k,D')\in\Cc$ has sum rank
\begin{equation*}
\begin{split}
m_1\rk(B+D_1)+\sum_{j=2}^k m_j\rk(D_j)+\mathrm{wt}(D')\geq \\ m_1(\maxrk(\mathcal{B})+1)+\sum_{j=2}^k m_j\maxrk(\Cc_j)+t-1= \\ \dim(\mathcal{C})+m_1-1>\dim(\Cc),
\end{split}
\end{equation*}
where $\mathrm{wt}(D')$ denotes the Hamming weight of $D'$, and the inequality follows from the assumption that $m_1>1$. This contradicts the assumption that $\Cc$ is an optimal anticode, completing the proof of the claim and of the theorem.
\end{proof}

\begin{example}
Denote by $\mathrm{rowsp}(M)$ the row-space of a matrix $M$. The optimal anticodes in $\F_q^{5\times 2}\times\F_q^{4\times 3}$ are exactly the codes of the form $$\{(A,B)\mid \mathrm{rowsp}(A)\subseteq U, \mathrm{rowsp}(B)\subseteq V\}$$ for some $U\subseteq\F_q^2$, $V\subseteq\F_q^3$ vector subspaces.
\end{example}

The next result is an easy consequence of Theorem \ref{thm:OAC}.

\begin{corollary}\label{cor:OAC}
Assume that either $q\neq2$ or $m_{\ell-2} \geq 2$. An $\F_q$-linear space $\Cc\subseteq\MM$ is an optimal anticode if and only if for all $i\in[\ell]$ there is $\Cc_i\subseteq\mathbb{F}_q^{m_i\times n_i}$ optimal anticode such that $\Cc=\prod_{i=1}^\ell\Cc_i$. 
\end{corollary}

\begin{proof}
By Theorem \ref{thm:OAC} $\Cc=\prod_{i=1}^k\Cc_i\times\Cc^\prime$, where $\Cc^\prime\subseteq\F_q^{\ell-k}$ is an optimal anticode, $k=\max\{i\in[\ell]\mid m_i>1\}$, and $\Cc_i\subseteq\F_q^{m_i\times n_i}$ are optimal anticodes for all $i\in[k]$.
If $q\neq 2$, then $\Cc^\prime$ is a product of zeroes and copies of $\F_q$ by \cite[Proposition 9]{Rav16}. If $q=2$ and $\ell-k\leq 2$, the same is true by direct inspection.  
\end{proof}

We conclude this section with a proof that the dual of an optimal anticode in the sum-rank metric is an optimal anticode, if $q \neq 2$ or $m_{\ell-2} >1$. 

\begin{proposition}\label{prop:dualOAC}
Let $q\neq 2$ or $ m_{\ell-2} > 1$. Then $\A \subseteq \matsumi$  is an optimal anticode if and only if $\A^\perp \subseteq \matsumi$ is an optimal anticode.
\end{proposition}

\begin{proof}
The dual of an optimal anticode in the rank-metric is an optimal anticode by~\cite[Theorem~54]{Rav16a}. 
The result now follows from Corollary~\ref{cor:OAC}, after observing that the dual of a product is the product of the duals.
\end{proof}

Notice that Corollary~\ref{cor:OAC} and Proposition~\ref{prop:dualOAC} cannot be extended to the case $q=2$ and $m_{\ell-2}=1$, since for $n\geq 3$ there exist optimal anticodes in $\F_2^n$ which are not products of zeroes and copies of $\F_2$, and whose dual is not an optimal anticode.

\begin{example}
Let $n\geq 3$ be odd and let $\Cc\subseteq\F_2^n$ be the even-weight code. Then $\Cc$ is an optimal anticode since $\dim(\Cc)=n-1=\maxwt(\Cc)$. Its dual $\Cc^\perp$ is the repetition code, which is not an optimal anticode since $\dim(\Cc^\perp)=1<n=\maxwt(\Cc^\perp)$.
\end{example}

\section{Isometries}\label{sec:isometries}

In this section we characterize the linear isometries of $\MM$ and use them to define a notion of equivalence between sum-rank metric codes. In the next section we define and study generalized weights and show that they are equivalence invariants. With our notion of equivalence, we also obtain that an optimal anticode is equivalent to a product of standard optimal anticodes in the rank metric.

\begin{definition}
An \textbf{$\F_q$-linear isometry} $\varphi$ in the sum-rank metric is an $\F_q$-linear homomorphism of $\MM$ such that $\srk(\varphi(C)) = \srk(C)$ for all $C\in \MM$. 
Two sum-rank metric codes $\Cc, \mathcal{D} \subseteq \matsumi$ are \textbf{equivalent} if there is an $\F_q$-linear isometry
 $\varphi : \MM \rightarrow \MM$ such that $\varphi(\Cc) = \mathcal{D}$. 
\end{definition}

Recall that every $\F_q$-linear isometry in the rank metric $\psi : \F_q^{m \times n} \rightarrow \F_q^{m \times n}$ has the form $\psi(A) = M A N$, or $\psi(A) = M A^t N$ if $m = n$, for some $M \in \mbox{GL}_m(\F_q)$ and $N \in \mbox{GL}_n(\F_q)$. We refer the interested reader to \cite{Hua,Wan} for a proof of this result. This allows us to characterize the $\F_q$-linear isometries in the sum-rank metric as follows.

\begin{theorem}\label{thm:isom}
Let $\varphi : \matsumi \longrightarrow \matsumi$ be an $\F_q$-linear isometry. 
Then there is a permutation
$$\sigma : [\ell] \longrightarrow [\ell]$$
with the property that $\sigma(i) = j$ implies $m_i=m_j$ and $n_i = n_j$ and there are rank-metric $\F_q$-linear isometries $\psi_i : \F_q^{m_i \times n_i} \longrightarrow \F_q^{m_i \times n_i}$  for $i\in[\ell]$ such that 
$$\varphi(C_1 , \dots , C_\ell) = (\psi_1(C_{\sigma(1)}), \dots , \psi_\ell(C_{\sigma(\ell)}))$$ for all  $(C_1, \dots , C_\ell) \in \matsumi$.
\end{theorem}

\begin{proof}
For $i\in[\ell]$, let $M_i=0\times\ldots\times 0\times\F_q^{m_i\times n_i}\times 0\times\ldots\times 0\subseteq\matsumi$ where the $i$th component is the only nonzero one. Let 
$\{(0,\ldots,0,E_{k,l}, 0,\ldots,0)\}_{(k,l)\in[m_i]\times[n_i]}$ 
be the standard basis of $M_i$. Then
$$\srk(\varphi(0,\dots , 0, E_{k,l}, 0 , \dots , 0))=1$$ 
for all $(k,l)\in [m_i]\times[n_i]$,
implying that $\varphi(0,\dots , 0, E_{k,l}, 0 , \dots , 0)$ has only one nonzero component for each choice of $k$ and $l$, say $i_{k,l}$. Further, we notice that for a given $k \in [m_i]$ 
\begin{equation}\label{row}
\srk\left(\varphi\left(0,\dots , 0, \sum_{l = 1}^{n_i}E_{k,l}, 0 , \dots , 0\right)\right) = 1, \end{equation}
and similarly for a given $l \in[n_i]$ we have that
\begin{equation} \label{column}
\srk\left(\varphi\left(0,\dots , 0, \sum_{k = 1}^{m_i}E_{k,l}, 0 , \dots , 0\right)\right) = 1. \end{equation}
By (\ref{row}) we have that
$$\srk\left(\sum_{l = 1}^{n_i} \varphi(0,\dots , 0, E_{k,l}, 0 , \dots , 0)\right) =1,$$
implying that $i_{k,l}$ does not depend on $k$. 
The same argument together with equation (\ref{column}) shows that $i_{k,l}$ does not depend on $l$ either.
It follows that for all $i$ there is a $j$ such that $\varphi(M_i)\subseteq M_j$. Since $\varphi^{-1}$ is a linear isometry, it follows from the same argument that that $\varphi^{-1}(M_j)\subseteq M_i$. Hence $\varphi(M_i)=M_j$. In particular, the map that sends $i$ to $j$ is a permutation of $[\ell]$, which we denote by $\sigma^{-1}$. Since $M_i$ and $M_j$ have the same weight distribution if and only if $n_i=\maxrk(M_i)=\maxrk(M_j)=n_j$ and
$m_i=\dim(M_i)/n_i=\dim(M_j)/n_j=m_j$. Therefore  
$$\begin{array}{ccccc}
 \varphi\mid_{M_i} & : &  M_i & \longrightarrow &M_j \\
 & & (0, \dots, 0, C_i , 0 , \dots , 0) &\longmapsto  &(0, \dots, 0, \psi_j(C_{i}) , 0 , \dots , 0)
\end{array}$$
for $j=\sigma^{-1}(i)$ and for some linear rank-metric isometry $\psi_j : \F_q^{m_j \times n_j} \rightarrow \F_q^{m_j \times n_j}$. 
Hence by linearity
$$\begin{array}{ccccc}
 \varphi& : &  \matsumi & \longrightarrow & \matsumi \\
 & & (C_1, \dots , C_\ell) &\longmapsto  &(\psi_1(C_{\sigma(1)}), \dots, \psi_\ell(C_{\sigma(\ell)})).
\end{array}$$
\end{proof}

The next corollary is immediate, after observing that every optimal anticode in the rank metric is equivalent to a standard optimal anticode, see e.g. \cite[Section 3]{G21}.

\begin{corollary}\label{stOAC}
For $i\in[\ell]$ let $\A_i\subseteq\F_q^{m_i\times n_i}$ be an optimal anticode and let $\A=\A_1\times\dots\times \A_\ell\subseteq\MM$. Then $\A$ is equivalent to $$\prod_{i=1}^\ell \langle E_{k,l}\mid k\in[m_i], l\in[u_i]\rangle,$$ where $u_i=\maxrk(\A_i)$.
\end{corollary}

It is natural to ask whether a result along the lines of the MacWilliams Extension Theorem holds in the sum-rank metric. It is clear that, since we do not have a MacWilliams Extension Theorem for rank-metric codes, we also cannot have a MacWilliams Extension Theorem for sum-rank metric codes. Moreover, in the sum-rank metric we have more pathologies than just those coming from the rank metric, as the next examples shows.

\begin{example}
Let $\ell=3$, $m_1=n_1=3,m_2=m_3=n_2=n_3=1$. Let $$\mathcal{C}=\left\{\left(\begin{pmatrix}
a & 0 & 0 \\ 0 & 0 & 0 \\ 0 & 0 & 0
\end{pmatrix},b,c\right) : a,b,c\in\F_q\right\}$$ and $$\mathcal{D}=\left\{\left(\begin{pmatrix}
a & 0 & 0 \\ 0 & b & 0 \\ 0 & 0 & c
\end{pmatrix},0,0\right) : a,b,c\in\F_q\right\}.$$ Then $\varphi:\Cc\rightarrow\mathcal{D}$ defined as $$\varphi\left(\begin{pmatrix}
a & 0 & 0 \\ 0 & 0 & 0 \\ 0 & 0 & 0
\end{pmatrix},b,c\right)=\left(\begin{pmatrix}
a & 0 & 0 \\ 0 & b & 0 \\ 0 & 0 & c
\end{pmatrix},0,0\right)$$ is an $\F_q$-linear isometry between $\Cc$ and $\mathcal{D}$, which does not extend to an $\F_q$-linear isometry of $\MM$ by Theorem~\ref{thm:isom}.
\end{example}

\begin{example}
Let $\ell=2$, $m_1=n_1=m_2=n_2=2$. Let $$\mathcal{C}=\left\{\left(\begin{pmatrix}
a & 0  \\ 0 & b 
\end{pmatrix},\begin{pmatrix}
c & 0  \\ 0 & d 
\end{pmatrix}\right) : a,b,c,d\in\F_q\right\}.$$ 
Then $\varphi:\Cc\rightarrow\mathcal{C}$ defined as $$\varphi\left(\begin{pmatrix}
a & 0  \\ 0 & b 
\end{pmatrix}, \begin{pmatrix}
c & 0  \\ 0 & d 
\end{pmatrix}\right)=\left(\begin{pmatrix}
a & 0  \\ 0 & c 
\end{pmatrix}, \begin{pmatrix}
b & 0  \\ 0 & d 
\end{pmatrix}\right)$$ is an $\F_q$-linear isometry between $\Cc$ and itself, which does not extend to an $\F_q$-linear isometry of $\MM$ by Theorem~\ref{thm:isom}.
\end{example}

\section{Generalized weights}\label{sec:generalizedweights}

In this section we define generalized weights in the sum-rank metric and establish some of their basic properties, including a weak monotonicity along the lines of the corresponding result for rank-metric codes. In addition, we prove that they satisfy Wei's Duality if $m_1=\ldots=m_\ell$. For general $m_i$'s, we show by means of an example that the generalized weights of a code do not determine those of its dual, hence Wei's Duality cannot hold. 


\begin{definition} 
Let $\Cc\subseteq \matsumi$ be a sum-rank metric code. For each $r\in[\dim(\Cc)]$, we define the  $r$-th generalized sum-rank weight of $\Cc$ as
\begin{equation*}
\begin{split}
d_r(\Cc) = \min\{\maxsrk(\A):\,&\A=\A_1\times\dots\times \A_\ell \mbox{ where }\A_i\subseteq\F_q^{m_i\times n_i} \\  & \mbox{are optimal anticodes and} \dim(\Cc \cap \A)\geq r\}.
\end{split}
\end{equation*}
\end{definition}

Notice that if $m_1=\dots=m_\ell=m$, then 
\begin{equation}
\begin{split}
d_r(\Cc) =\frac{1}{m} \min\{\dim(\A):\,&\A= \A_1\times\dots\times \A_\ell \text{ where } \A_i\subseteq\F_q^{m\times n_i} \\  & \mbox{are optimal anticodes and} \dim(\Cc\cap \A)\geq r\}.
\end{split}
\end{equation}

\begin{remark}
We could have defined $d_r(\Cc)$ to be
$$d'_r(\Cc)=\min\{\maxsrk(\A) : \, \A\text{ an optimal anticode and } \dim(\Cc\cap \A)\geq r \}.$$
    
For either $q \neq 2$ or $m_{\ell-2} >1$ we have that $d_r(\Cc)=d'_r(\Cc)$ as, by Corollary~\ref{cor:OAC}, $\A$ is an optimal anticode if and only if $\A=\A_1\times\cdots\times \A_\ell$ for $\A_i$ optimal anticode in $\F_q^{m_i\times n_i}$. In the case $q=2$ and $m_{\ell-2}=1$ one has
$$d'_r(\Cc)\leq d_r(\Cc).$$

Notice moreover that $d_r(\Cc)$ recovers the Hamming weights, since the cardinality of a support of a code is the minimum dimension of a code which contains it and is a product of copies of $\F_q$ and zeros. If $q=2$, then $d'_r(\Cc)$ does not recover the Hamming weights, as there are optimal binary anticodes which are not a product of copies of $\F_2$ and zeros. See also the example following Theorem~10 in~\cite{Rav16}.
\end{remark}

\begin{remark}
It follows from the definition that the generalized weights are invariant under code equivalence.
\end{remark}

As an example, we compute the generalized weights of optimal anticodes.

\begin{example}\label{genwtsOAC}
For $i\in[\ell]$ let $\A_i\subseteq\F_q^{m_i\times n_i}$ be an optimal anticode and let $\A= \A_1\times\dots\times \A_\ell\subseteq\MM$ with $\dim\A_i=m_iu_i$. By Corollary \ref{stOAC} and the previous remark, $d_r(\A)=d_r(\A')$ for $r\in[\dim(\Cc)]$, where $\A'=\prod_{i=1}^\ell \langle E_{k,l}\mid (k,l)\in[m_i]\times[u_i]\rangle$. Let $j\in[\ell]$, $0\leq \delta\leq u_j-1$, $r=\sum_{i=1}^{j-1}m_iu_i+m_j\delta$. Then
$$d_{r+1}(\A)=\ldots=d_{r+m_j}(\A)=u_1+\ldots+u_{j-1}+\delta+1.$$
\end{example}

\begin{lemma}\label{lemma:nos}
Let $m_1\geq ...\geq m_\ell\in \mathbb{N}$, $u_1,\ldots,u_\ell$, $u'_1,\ldots,u'_\ell\in\mathbb{R}_{\geq 0}$ such that $\sum_{i=1}^\ell u_i=\sum_{i=1}^\ell u'_i$ and such that there exists $k$ with $u_i\geq u'_i$ for all $1\leq i\leq k$ and $u_i\leq u'_i$ for all $k<i\leq\ell$, then $\sum_{i=1}^\ell m_iu_i\geq\sum_{i=1}^\ell m_i u'_i$.
\end{lemma}

\begin{proof}
Since $\sum_{i=1}^k( u_i-u'_i)=\sum_{i=k+1}^\ell (u'_i-u_i)$ and $m_1\geq ...\geq m_\ell$, then $$\sum_{i=1}^k m_i(u_i-u'_i)\geq m_k\sum_{i=1}^k(u_i-u'_i)\geq m_{k+1}\sum_{i=k+1}^\ell(u'_i-u_i)\geq\sum_{i=k+1}^\ell m_i (u'_i-u_i),$$ which proves the thesis.
\end{proof}

In the next proposition we establish some basic properties of generalized weights.
Notice that in the case $m_1=\ldots=m_\ell$ one gets inequalities of the same form as those in~\cite[Theorem 30]{Rav16}.

\begin{proposition}\label{properties} 
Let $0\neq\Cc \subseteq\mathcal{D} \subseteq \MM$, then:
\begin{enumerate}
\item $d_1(\Cc)= d(\Cc)$,
\item $d_r(\Cc)\leq d_s(\Cc)$ for $1\leq r\leq s\leq\dim(\Cc)$,
\item $d_r(\Cc)\geq d_r(\mathcal{D})$ for $r\in[\dim(\Cc)]$,
\item $d_{\dim(\Cc)}(\Cc)\leq n_1 + \dots + n_\ell$,
\item $d_{r+n_1m_1+\cdots+n_{j-1}m_{j-1}+\delta m_{j}}(\Cc)\geq d_r(\Cc)+n_1+\cdots+n_{j-1}+\delta$\\ 
for $j\in[\ell]$, $r\in[ \dim(\Cc)-(n_1m_1+\cdots+n_{j-1}m_{j-1}+\delta m_j)]$, and $0\leq\delta\leq n_j-1$.
\end{enumerate}
\end{proposition}

\begin{proof} 
1. Let $C=(C_1,\ldots,C_\ell)\in \Cc$ be an element of minimum sum-rank. Let $\A_i$ be an optimal anticode of $\dim(\A_i)=m_i\rk(C_i)$ containing $C_i$ and let $\A=\A_1\times\cdots\times \A_\ell$. Then $\Cc \cap \A \neq 0$, hence $d_1(\Cc)\leq d(\Cc)$. To prove that they are equal, observe that if $\A'$ is an optimal anticode with $\maxsrk(\A')<d(\Cc)$, then $\A'\cap\Cc=0$. 

2., 3., and 4. follow directly from the definition.

5. Let $s=r+n_1m_1+\cdots+n_{j-1}m_{j-1}+\delta m_j$.
Let $\A = \A_1\times\cdots\times \A_\ell$ be an optimal anticode such that $\dim(\Cc \cap \A)\geq s$ and $d_s(\Cc)=\maxsrk(\A)$. For $i\in[\ell]$, write $\dim(\A_i)=m_i u_i$. Since $$\sum_{i=1}^{\ell} m_iu_i=\dim(\A)\geq \dim(\Cc \cap \A)\geq s> n_1m_1+\cdots+n_{j-1}m_{j-1}+\delta m_j$$ and $m_1\geq\cdots\geq m_\ell$, then $d_s(\Cc)=u_1+\ldots+u_\ell> n_1+\cdots+n_{j-1}+\delta$ by Lemma~\ref{lemma:nos}.
Let $v_1,\ldots,v_{\ell}$ be such that $n_1+\cdots+n_{j-1}+\delta= v_1+\cdots+v_\ell$ and $v_i\leq u_i$ for $i\in[\ell]$.
We have that $n_1m_1+\cdots+n_{j-1}m_{j-1}+\delta m_j\geq v_1m_1+\cdots+v_\ell m_\ell$, since $m_1\geq\cdots\geq m_\ell$. For all $i\in[\ell]$ there exist optimal anticodes $\A'_i\subseteq \A_i$ of $\dim(\A'_i)=m_i(u_i-v_i)$. Let $\A'=\A'_1\times\cdots\times \A'_\ell$, then
\begin{align*}
\dim(\Cc \cap \A')&\geq s-(v_1m_1+\cdots+v_\ell m_ \ell)\\
&\geq s-(n_1m_1+\cdots+n_{j-1}m_{j-1}+\delta m_j)\\
&=r
\end{align*}
hence
\begin{equation*}
d_r(\Cc)\leq \sum_{i=1}^\ell (u_i-v_i) = d_s(\Cc)-(n_1+\cdots+n_{j-1}+\delta).\qedhere
\end{equation*}
\end{proof}

From parts 4. and 5. of Proposition \ref{properties}, we easily obtain the following Singleton-type bound. This bound will be improved in Theorem \ref{singletonbound}.

\begin{corollary}\label{cor:singletontypebound}
Let $j\in[\ell]$, $0\leq\delta\leq n_j-1$, $0\leq s\leq m_j-1$, and let $\Cc\subseteq \MM$ be a non-trivial code of $$\dim(\Cc)=\sum_{i=1}^{j-1}m_in_i+\delta m_j+s.$$ Then $$d(\Cc)\leq \sum_{i=j}^{\ell} n_i-\delta+\left\{\begin{array}{ll}
1 & \mbox{ if } s=0 \\
0 & \mbox{ else.}
\end{array}\right.$$
\end{corollary}

The next lemma will be useful in Section \ref{sec:MRD} for computing the generalized weights of an MSRD code.

\begin{lemma}\label{lemma}
Let $\Cc\subseteq\MM$ be a code and let $k\in[\ell]$, $r+m_k\in[\dim(\Cc)]$. If
$$d_{r+m_k}(\Cc)>\sum_{i=1}^{k-1}n_i$$ then
$$d_{r+m_k}(\Cc)\geq d_r(\Cc)+1.$$
\end{lemma}

\begin{proof}
Let $\mathcal{A}=\mathcal{A}_1\times\cdots\times\mathcal{A}_\ell$ be an optimal anticode such that $\maxsrk(\mathcal{A})=d_{r+m_k}(\Cc)$ and $\dim(\Cc\cap\mathcal{A})\geq r+m_k$. We claim that there exists $k\leq j\leq \ell$ such that $\mathcal{A}_j\neq 0$. In fact, if this were not the case, then 
$$\sum_{i=1}^{k-1}n_i\geq\maxsrk(\mathcal{A})=d_{r+m_k}(\Cc).$$
Let $\mathcal{A}'\subseteq\mathcal{A}$ be an optimal anticode such that
$$\dim(\mathcal{A}')=\dim(\mathcal{A})-m_j \;\mbox{ and }\; \maxsrk(\A')=\maxsrk(\A)-1.$$
One has
$$\dim(\Cc\cap\mathcal{A}')\geq\dim(\Cc\cap\mathcal{A})-m_j\geq r+m_k-m_j\geq r,$$
hence
$$d_r(\Cc)\leq\maxsrk(\A')=d_{r+m_k}(\Cc)-1.$$
\end{proof}

The next theorem extends Wei's Duality Theorem~\cite[Theorem~3]{Wei} and~\cite[Corollary 38]{Rav16}.
Let $m_1=\ldots=m_\ell=m$ and let $\Cc \subseteq \MM$ be a sum-rank metric code. For any $r\in \Z$ define 
$$W_r(\Cc) = 
\{ d_{r+sm}(\Cc) : s \in \Z , r+sm\in[\dim(\Cc)]\},$$
$$\overline{W}_r(\Cc) = \bigg\{ n+ 1 - d_{r+sm}(\Cc) : s \in \Z , r+sm\in[\dim(\Cc)]\bigg \}.$$

The same arguments as in \cite[Corollary 38]{Rav16} together with Proposition~\ref{properties} prove the next theorem.

\begin{theorem} \label{weiduality}
Let $m_1=\ldots=m_\ell=m$, $r\in[m]$, and let $\Cc \subseteq \MM$ be a sum-rank metric code. Then 
$$W_r(\Cd) = [n] \backslash \overline{W}_{r + \dim(\Cc)}(\Cc).$$
In particular the generalized weights of a sum rank metric code $\Cc$ determine the generalized weights of $\Cd$.
\end{theorem}

The next example shows that the generalized weights of a code do not determine those of its dual for arbitrary $m_i$'s.

\begin{example}
Let $\Cc_1,\Cc_2\subseteq\F_2^{3\times 1}\times \F_2^{2\times 2}$ be given by
\begin{equation*}
    \begin{split}
        \Cc_1&=0\times \F_2^{2\times 2}\\ 
        \Cc_2&=\left\{\left( \begin{pmatrix} a \\ b \\ 0\end{pmatrix}, \begin{pmatrix} c & d \\0 &0\end{pmatrix}\right) : (a,b,c,d) \in \mathbb{F}_2^4\right\}.
    \end{split}
\end{equation*}
One can check that $d_1(\Cc_i) = d_2(\Cc_i) = 1$ and $d_3(\Cc_i) = d_4(\Cc_i) = 2 $ for $i = 1,2$. The corresponding duals 
\begin{equation*}
    \begin{split}
        \Cc_1^\perp&=\F_2^{3 \times 1}\times 0 \\ 
        \Cc_2^\perp&=\left\{\left( \begin{pmatrix}0\\ 0 \\ a\end{pmatrix}, \begin{pmatrix}0 &0\\ b & c \end{pmatrix}\right) : (a,b,c) \in \mathbb{F}_2^3\right\}
    \end{split}
\end{equation*}
have different generalized weights, as $d_3(\Cc_1^\perp) = 1$ and $d_3(\Cc_2^\perp) = 2$. 
\end{example}

\begin{remark}
Notice that the first code in the previous example is an optimal anticode, while the second one is not, as its first component is not an optimal rank-metric anticode. Therefore, the example also shows that in the sum-rank metric there exist codes which have the same dimension and generalized weights as an optimal anticode, without being one. This is in contrast with codes endowed with the rank metric or the Hamming metric, where a code which has the same dimension and generalized weights as an optimal anticode is an optimal anticode.
\end{remark}

\begin{remark}
There is another simple situation in which the generalized weights of the dual code are determined by numerical data on the original code. 
Let $\Cc=\Cc_1\times\ldots\times\Cc_\ell$, then the generalized weights of $\Cc$ satisfy 
\begin{equation*}
d_r(\Cc)=\min\left\{\sum_{i=1}^{\ell}d_{r_i}(\Cc_i) : \sum_{i=1}^{\ell} r_i=r, r_i\in[\dim(\Cc_i)]\right\}.
\end{equation*}
The generalized weights of the rank-metric codes $\Cc_1,\ldots,\Cc_\ell$ determine those of $\Cc_1^\perp,\ldots,\Cc_\ell^\perp$, hence they determine the generalized weights of $\Cc^\perp$.
\end{remark}

We conclude this section with a result on the weights of a code which is $\F_{q^m}$-linear or, more generally, $\F_{q^k}$-linear. 
Let $k=\gcd\{m_1,\ldots,m_\ell\}$. As $k\mid m_i$ for all $i\in[\ell]$, then $\F_{q^{m_1}}^{n_1}\times\dots\times\F_{q^{m_{\ell}}}^{n_{\ell}}$ is a vector space over $\F_{q^{k}}$.
For $i\in[\ell]$, let $\Gamma_{i}=\{\gamma_{1,i},\dots,\gamma_{m_i,i}\}$ be a basis of $\F_{q^{m_i}}$ over $\F_q$. For every $w\in\F_{q^{m_i}}^{n_i}$ define $\Gamma_i(w)\in\F_{q}^{m_i\times n_i}$ via the identity
\begin{equation*}
   \begin{pmatrix} \gamma_{1,i} & \dots & \gamma_{m_i,i}\end{pmatrix} \Gamma_i(w)=w.
\end{equation*}
For every $v=(v_1,\dots,v_{\ell})\in\F_{q^{m_1}}^{n_1}\times\dots\times\F_{q^{m_{\ell}}}^{n_{\ell}}$, define $\Gamma(v)\in\MM$ as
\begin{equation*}
    (\Gamma(v))_i=\Gamma_i(v_i).
\end{equation*}
Let $\mathcal{V}\subseteq\F_{q^{m_1}}^{n_1}\times\dots\times\F_{q^{m_{\ell}}}^{n_{\ell}}$ be a vector space over $\F_{q^{k}}$. The set $\Gamma(\mathcal{V})=\{\Gamma(v):v\in \mathcal{V}\}$ is the sum-rank metric code associated to $\mathcal{V}$ with respect to $\{\Gamma_1,\dots,\Gamma_{\ell}\}$. We say that $\Gamma(\mathcal{V})$ is $\F_{q^{k}}$-linear, see also \cite[Definition 11.1.3]{G21}. 
In the next theorem we extend the result in~\cite[Theorem~28]{Rav16} to the sum-rank metric case. 
The statement in particular applies to $\F_{q^m}$-linear codes in the case when $m_1=\ldots=m_\ell=m$.

\begin{theorem}\label{thm:klinear}
Let $k=\gcd\{m_1,\ldots,m_\ell\}$, let $\mathcal{V}\subseteq\F_{q^{m_1}}^{n_1}\times\dots\times\F_{q^{m_{\ell}}}^{n_{\ell}}$ be an $\F_{q^{k}}$-linear vector space with $\dim_{\F_{q^{k}}}(\mathcal{V})=t$. If $m_i>n_i$ for $i\in[\ell]$, then $$d_{kr+1}(\Gamma(\mathcal{V}))=\ldots=d_{k(r+1)}(\Gamma(\mathcal{V}))$$ for $0\leq r<t$.
\end{theorem}

\begin{proof}
Write $\Cc$ for $\Gamma(\mathcal{V})$. By Proposition~\ref{properties}, $d_{kr+1}(\Cc)\leq\ldots\leq d_{k(r+1)}(\Cc)$. Therefore it suffices to show that $d_{kr+1}(\Cc)=d_{k(r+1)}(\Cc)$. Since $m_i>n_i$ for $i\in[\ell]$, $\mathcal{A}$ is an $\F_{q^{k}}$-linear code and so $\Cc\cap\mathcal{A}$ is $\F_{q^{k}}$-linear too. Since the dimension over $\F_q$ of an $\F_{q^{k}}$-linear vector space is divisible by $k$, if $\dim(\Cc\cap\mathcal{A})\geq kr+1$, then $\dim(\Cc\cap\mathcal{A})\geq k(r+1)$. Therefore we conclude that $d_{kr+1}(\Cc)\geq d_{k(r+1)}(\Cc)$.
\end{proof}

\begin{remark}
Although the condition that $m>n$ is missing in the statement of \cite[Theorem~28]{Rav16}, it is necessary for the result to hold. In fact, \cite[Example 6.15]{gluesing2021q} is a counterexample to the statement of \cite[Theorem~28]{Rav16} for square matrices.
\end{remark}

\section{MSRD codes}\label{sec:MRD}

In this section we define MSRD and $r$-MSRD codes, and compute their generalized weights. 

\begin{notation}
Let $\mu\in[n]$. We denote by $\mathbb{A}(\mu)$ the set of optimal anticodes of the form $\A=\A_1\times\ldots\times\A_\ell\subseteq\MM$, with $\A_i\subseteq\F_q^{m_i\times n_i}$ optimal rank-metric anticode for all $i\in[\ell]$ and $\maxsrk(\A)=\sum_{i=1}^\ell \maxrk(A_i)=\mu$.
\end{notation}

The next result follows from Lemma~\ref{lemma:nos}.

\begin{lemma}\label{dim oac}
Let $\mu\in[n]$ and write $\mu=\sum_{i=1}^{j-1} n_i+\delta=\sum_{i=l+1}^\ell n_i+\delta'$ for some $j,l\in[\ell]$, $\delta\in[n_j]$, and $\delta^\prime\in[n_l]$. Then
$$\min_{\A\in\mathbb{A}(\mu)}\dim(\A)=\sum_{i=l+1}^\ell m_in_i+\delta' m_l$$ 
and
$$\max_{\A\in\mathbb{A}(\mu)}\dim(\A)=\sum_{i=1}^{j-1} m_in_i+\delta m_j.$$
Moreover, if $$\min_{\A\in\mathbb{A}(\mu)}\dim(\A)=\max_{\A\in\mathbb{A}(\mu)}\dim(\A),$$ then either $\mu=n$ or $m_1=\ldots=m_\ell$.
\end{lemma}

\begin{notation}
Let $\mu\in[n]$ and write $\mu=\sum_{i=1}^{j-1}n_i+\delta+1$, $0\leq \delta\leq n_j-1$.
Throughout the section, we denote $$r_\mu=\max_{\mathcal{A}\in\mathbb{A}(\mu)}\dim(\mathcal{A})=\sum_{i=1}^{j-1} m_in_i+(\delta+1)m_j.$$ 
\end{notation}

The Singleton Bound for rank-metric codes was first proved in~\cite[Theorem~5.4]{Del}. A Singleton Bound for sum-rank metric codes was established in~\cite[Theorem 3.2]{BGRMSRD}, for codes which are not necessarily linear. Our next theorem generalizes the previous results in the case of linear sum-rank metric codes.

\begin{theorem}\label{singletonbound}
Let $\Cc\subseteq\MM$ be a code and let $r\in[\dim(\Cc)]$. Let $j\in[\ell]$ and $0\leq\delta\leq n_j-1$ be such that $$d_r(\Cc)-1\geq\sum_{i=1}^{j-1} n_i+\delta.$$
Then
\begin{equation}\label{sb}
\dim(\Cc)\leq\sum_{i=j}^\ell m_in_i-m_{j}\delta+r-1.
\end{equation}
\end{theorem}

\begin{proof}
Let $\mathcal{A}_i=\mathbb{F}_q^{m_i\times n_i}$ for $i\in[j-1]$, let $\mathcal{A}_j\subseteq\mathbb{F}_q^{m_j\times n_j}$ be an optimal anticode of dimension $\delta m_j$, and let $\A_i=0$ for $j+1\leq i\leq\ell$.
Let $\mathcal{A}=\mathcal{A}_1\times\cdots\times\mathcal{A}_\ell$, then
$$\dim(\Cc\cap\mathcal{A})\leq r-1.$$
Therefore
\begin{align*}
\dim(\Cc)+\sum_{i=1}^{j-1} m_in_i+m_j \delta-r+1&\leq\dim(\Cc)+\dim(\mathcal{A})-\dim(\Cc\cap\mathcal{A})\\
&=\dim(\Cc+\mathcal{A})\leq\sum_{i=1}^\ell m_in_i.
\end{align*}
\end{proof}

Theorem \ref{singletonbound} yields upper bounds on all the generalized weights of $\Cc$.

\begin{corollary}\label{cor.rsb}
Let $\Cc\subseteq\MM$ be a code and let $r\in[\dim(\Cc)]$, $j\in[\ell]$, and $0\leq\delta\leq n_j-1$ be such that $\dim(\Cc)\geq\sum_{i=j}^\ell m_in_i-m_{j}\delta+r$. Then
$$d_r(\Cc)\leq\sum_{i=1}^{j-1} n_i+\delta.$$
In particular, if $\dim(\Cc)=\sum_{i=j}^\ell m_in_i-m_{j}\delta$, then $$d_1(\Cc)\leq\ldots\leq d_{m_j}(\Cc)\leq \sum_{i=1}^{j-1} n_i+\delta+1.$$
\end{corollary}

Corollary \ref{cor.rsb} suggests the following definition of MSRD code. The same definition was given in \cite[Definition 3.3]{BGRMSRD} for codes which are not necessarily linear.

\begin{definition}
A code $\Cc$ is MSRD if there exist $j\in[\ell]$ and $0\leq \delta\leq n_j-1$ such that $$d(\Cc)=\sum_{i=1}^{j-1} n_i+\delta+1 \;\mbox{ and }\; \dim(\Cc)=\sum_{i=j}^{\ell} m_in_i-\delta m_{j}.$$
\end{definition}

Next we study some properties which are closely related to being MSRD.

\begin{enumerate}
\item[$(C0)$] For any $\mathcal{A}$ optimal anticode of $\maxsrk(\A)=d(\Cc)-1$ and $\dim(\A)=r_{d(\Cc)-1}$ one has $\Cc+\A=\MM$.
\item[$(C1)$] The code $\Cc$ has $\dim(\Cc)=\sum_{i=j}^\ell m_in_i-m_j\delta$ and for any $\mathcal{A}$ optimal anticode of $\maxsrk(\A)\leq \sum_{i=1}^{j-1}n_i+\delta$ one has $\Cc\cap\A=0$. 
\item[$(C2)$] For any $\mathcal{A}\in\mathbb{A}(d(\Cc))$, let  $k=\max\{i\in[\ell]\ |\ \A_i\neq 0\}$. Then
$$\dim(\Cc\cap\mathcal{A})\geq m_k.$$
\item[$(C3)$] The code $\Cc$ has $d(\Cc)+d(\Cc^\perp)=n+2$.
\end{enumerate}

It is clear that being MSRD is equivalent to satisfying $(C0)$. We now show that it is also equivalent to satisfying $(C1)$.

\begin{proposition}
Let $j\in[\ell]$ and $0\leq\delta\leq n_j-1$. Let $0\neq\Cc\subseteq\MM$ be a code. 
Then $\Cc$ is MSRD if and only if it satisfies $(C1)$.
\end{proposition}

\begin{proof}
Suppose that $\Cc$ is MSRD of $\dim(\Cc)=\sum_{i=j}^{\ell} m_in_i-\delta m_{j}$. Let $\mathcal{A}$ be an optimal anticode of $\maxsrk(\A)\leq d(\Cc)-1$. Then $\Cc\cap\A=0$ since, for every $0\neq C\in\Cc$, one has $\srk(C)\geq d(C)>\maxsrk(\A)$, so $C\not\in\A$. 

Suppose now that $\Cc$ satisfies $(C1)$. 
Then $d(\Cc)\leq \sum_{i=1}^{j-1} n_i+\delta+1$ by Corollary \ref{cor.rsb}. Let $C=(C_1,\ldots,C_\ell)\in\Cc$. For each $i\in[\ell]$, there is an optimal rank-metric anticode $\A_i\subseteq\F_q^{m_i\times n_i}$ of $\dim(\A_i)=m_i\rk(C_i)$ which contains $C_i$. Therefore $\A=\A_1\times\ldots\times\A_\ell$ is an optimal sum-rank metric anticode of $\maxsrk(\A)=\srk(C)$ which contains $C$. Since $\Cc\cap\A\neq 0$, it must be that $\maxsrk(\A)=\srk(C)\geq \sum_{i=1}^{j-1}n_i+\delta+1$, therefore $\Cc$ is MSRD. 
\end{proof}

\begin{proposition}\label{prop.moac}
Let $0\neq\Cc\subseteq\MM$ be a code and write its minimum distance as $d=d(\Cc)=\sum_{i=1}^{j-1}n_i+\delta+1$, where $j\in[\ell]$ and $0\leq\delta\leq n_j-1$. For $S\subseteq[n]$, 
denote by $\F_q[S]$ the set of elements of $\MM$ which are
zero outside of the columns indexed by $S$.
For any $d\leq h\leq n$, let $S_h:=[d-1]\cup\{h\}$. The following hold:
\begin{enumerate}
\item $\Cc$ is MSRD if and only if for any $d\leq h\leq n$ we have $$\dim(\Cc\cap\mathbb{F}_q[S_h])=m_{k}$$
where $k=\max\{\nu\ |\ \sum_{i=1}^{\nu-1} n_i<h\}$.
\item If $\Cc$ satisfies $(C2)$, then $\Cc$ is MSRD.
\end{enumerate}
\end{proposition}

\begin{proof} 
1. Assume that $\Cc$ is MSRD and let $d\leq h\leq n$. We have
\begin{align*}
\dim(\Cc\cap\F_q[S_h])&\geq \dim (\Cc)+\dim(\F_q[S_h])-\sum_{i=1}^\ell m_in_i\\
&=\sum_{i=j}^\ell m_in_i-\delta m_j+\sum_{i=1}^{j-1} m_in_i+\delta m_j+m_{k}-\sum_{i=1}^\ell m_in_i\\
&= m_{k}.
\end{align*}
    
Conversely, suppose that for $d\leq h\leq n$ one has $\dim (\Cc\cap\F_q[S_h])\geq m_k$. Let $d\leq h'\leq n$, $h\neq h'$. Then 
$$\dim (\Cc\cap \F_{q}[S_h]\cap\F_q[S_{h'}])=\dim (\Cc\cap\F_q[[d-1]])=0$$
hence
\begin{equation}\label{eqn:dimC}
\dim(\Cc)\geq\sum_{h=d}^n \dim (\Cc\cap\F_q[S_h])\geq \sum_{i=j}^{\ell} m_in_i-\delta m_j.
\end{equation}
Theorem \ref{singletonbound} gives the reverse inequality, hence $\Cc$ is MSRD.

This proves that $\Cc$ is MSRD if and only if $\dim (\Cc\cap\F_q[S_h])\geq m_k$ for all $d\leq h\leq n$.
Notice moreover that (\ref{eqn:dimC}) and Theorem \ref{singletonbound} imply that, if $\dim (\Cc\cap\F_q[S_h])\geq m_k$ for all $d\leq h\leq n$, then in fact $\dim (\Cc\cap\F_q[S_h])=m_k$ for all $d\leq h\leq n$. This concludes the proof of the first part of the statement.
    
2. Suppose that $\Cc$ satisfies (C2). For any $d\leq h\leq n$, letting $\A=\F_q[S_h]\in\mathbb{A}(d)$, one has that $\dim (\Cc\cap \F_{q}[S_h])\geq m_k$. As shown in 1., combining (\ref{eqn:dimC}) and Theorem \ref{singletonbound} one obtains that $\Cc$ is MSRD.
\end{proof}

The next examples show that there exist nontrivial codes which satisfy property $(C2)$ and that not every MSRD code satisfies $(C2)$.

\begin{example}
In $\mathbb{F}_2^{2\times 2}\times\mathbb{F}_2^{1\times 1}$, let
$$C=\left\langle\left(\begin{pmatrix} 1&0\\0&0\end{pmatrix},1\right),\left(\begin{pmatrix} 0&0\\0&1\end{pmatrix},1\right),\left(\begin{pmatrix} 0&1\\1&0\end{pmatrix},1\right)\right\rangle.$$
We have $d(C)=2$ and $C$ satisfies $(C2)$.
\end{example}


\begin{example}
Let $\mathcal{C}\subseteq\mathbb{F}_2^{3\times 3}\times\mathbb{F}_2^{2\times 2}\times\mathbb{F}_2\times\mathbb{F}_2\times\mathbb{F}_2$ be
$$\Cc=\left\langle\left(\begin{pmatrix} 1&0&0\\ 0&1&0\\ 0&0&1\end{pmatrix},\begin{pmatrix} 1&0\\ 0&1\end{pmatrix},1,1,0\right),\left(\begin{pmatrix}0&0&1\\ 1&0&1\\ 0&1&0\end{pmatrix}, \begin{pmatrix} 0&1\\ 1&1\end{pmatrix}, 0,1,1\right)\right\rangle.$$
The code $\mathcal{C}$ has dimension 2 with $d(\mathcal{C})=7$, hence it is an MSRD code. Consider now the optimal anticode
 $$\mathcal{A}=\langle E_{i,1},E_{i,2}\ |\ i\in[3]\rangle\times\mathbb{F}_2^{2\times 2}\times\mathbb{F}_2\times\mathbb{F}_2\times\mathbb{F}_2.$$
We have $\maxsrk(\mathcal{A})=7$ and 
$\mathcal{A}\cap\mathcal{C}=0$. Hence $\mathcal{C}$ does not satisfy $(C2)$.
\end{example}

\begin{proposition}\label{prop:C3}
Let $\Cc\subseteq\MM$ be a non-trivial code.
Then $\Cc$ satisfies $(C3)$ if and only if both $\Cc$ and $\Cc^\perp$ are MSRD.
\end{proposition}

\begin{proof}
Write $\dim(\Cc)=\sum_{i=j}^{\ell}m_in_i-\delta m_j-s$ for some $j\in[\ell]$, $0\leq\delta\leq n_j-1$, and $0\leq s\leq m_j-1$. By Corollary \ref{cor.rsb} \begin{equation}\label{eqn:dC}
d_1(\Cc)\leq\sum_{i=1}^{j-1} n_i+\delta+1.
\end{equation}
Moreover, $\dim(C^\perp)=\dim(\MM)-\dim(\Cc)=\sum_{i=1}^{j-1} m_in_i+\delta m_{j}+s$, which by Corollary \ref{cor:singletontypebound} implies that 
\begin{equation}\label{eqn:dCdual}
    d_1(\Cc^\perp)\leq\sum_{i=j}^{\ell} n_j-\delta+\left\{\begin{array}{ll}
    1 & \mbox{ if } s=0 \\
    0 & \mbox{else.}
\end{array}\right.
\end{equation}
Therefore 
$$d(\Cc)+d(\Cc^\perp)\leq \left\{\begin{array}{ll}
    n+2 & \mbox{ if } s=0 \\
    n+1 & \mbox{else.}
\end{array}\right.$$
If $\Cc$ satisfies $(C3)$, then $s=0$ and both $\Cc$ and $\Cc^\perp$ are MSRD. Conversely, if $\Cc$ and $\Cc^\perp$ are MSRD, then $s=0$ and both (\ref{eqn:dC}) and (\ref{eqn:dCdual}) are equalities. It follows that $\Cc$ satisfies $(C3)$.
\end{proof}

In the next proposition we prove that, if $m_1=\ldots=m_\ell$, then properties $(C2)$ and $(C3)$ are equivalent to being MSRD.

\begin{proposition}\label{meq}
Let $\Cc\subseteq\MM$ be a non-trivial code.
If $m_1=\ldots=m_\ell=m$, then both $(C2)$ and $(C3)$ are equivalent to being MSRD. In particular, the dual of an MSRD code is MSRD.
\end{proposition}

\begin{proof}
Let $\Cc\subseteq\MM$ be a non-trivial code.
If $\Cc$ is MSRD, then it satisfies $(C3)$ by \cite[Theorem 6.1]{BGRMSRD}. If $\Cc$ satisfies property $(C3)$, then it is MSRD by Proposition \ref{prop:C3}.

If $\Cc$ satisfies $(C2)$, then it is MSRD by Proposition \ref{prop.moac}.
We now prove that if $\Cc$ is MSRD, then it satisfies $(C2)$.
Let $\mathcal{A}\in\mathbb{A}(d(\Cc))$, then 
$$\dim(\Cc)+\dim(\mathcal{A})\leq mn+\dim(\Cc\cap\mathcal{A}).$$
Hence by Lemma \ref{dim oac} we have
$$mn+m\leq mn+\dim(\Cc\cap\mathcal{A}),$$
so $\Cc$ satisfies $(C2)$. 
\end{proof}

Moreover, one can prove that $(C3)$ defines a trivial family of codes, unless $m_1=\cdots=m_\ell$. Notice that this shows in particular that the dual of a non-trivial MSRD code can never be MSRD, unless $m_1=\cdots= m_\ell$.

\begin{proposition}
If there exists a non-trivial code $\Cc\subseteq\MM$ that satisfies $(C3)$, then $m_1=\cdots= m_\ell$.
\end{proposition}

\begin{proof}
Write $d(\Cc^\perp)-1=\sum_{i=1}^{k-1} n_i+\varepsilon$ for some $k\in[\ell]$ and $0\leq\varepsilon\leq n_k-1$. Since $d(\Cc)+d(\Cc^\perp)-2=n$, one has 
\begin{equation}\label{eq:dmin}
d(\Cc)-1=\sum_{i=1}^{j-1} n_i+\delta=\sum_{i=k}^{\ell} n_i-\varepsilon
\end{equation}
for some $j\in[\ell]$ and $0\leq\delta\leq n_j-1$.
Since $\Cc$ and $\Cc^\perp$ are MSRD by Proposition \ref{prop:C3}, one has 
\begin{equation}\label{eq:dim}
\dim(\Cc)=\sum_{i=j}^{\ell} n_im_i-\delta m_j=\sum_{i=1}^{k-1} n_im_i+\varepsilon m_k=\dim(\MM)-\dim(\Cc^\perp).
\end{equation}
Lemma \ref{dim oac}, together with (\ref{eq:dim}), implies that
$$\max\dim\mathbb{A}(d(\Cc)-1)=\min\dim\mathbb{A}(d(\Cc)-1),$$ which by Lemma \ref{dim oac} implies that $m_1=\cdots=m_\ell$.
\end{proof}

In the remainder of this section, we study the generalized weights of MSRD codes and propose a definition of $r$-MSRD codes, analogous to that of $r$-MRD codes. 
The next theorem states that the generalized weights of an MSRD code are determined by its parameters. This generalizes similar results for MDS codes in the Hamming metric and MRD codes in the rank metric. We postpone the proof, since in Theorem \ref{thm:rMSRD} we will prove a more general result. 

\begin{theorem}\label{genwtsMSRD}
Let $\Cc\subseteq\MM$ be an MSRD code and write $d(\Cc)=\sum_{i=1}^{j-1} n_i+\delta+1$ for some $j\in[\ell]$ and $0\leq \delta\leq n_j-1$. Let 
$d(\Cc)\leq h\leq n$ and let $k=\max\{\nu\mid\sum_{i=1}^{\nu-1}n_i<h\}$.
Let $r\in[\dim(\Cc)]$ be of the form
$$r=r_h-r_{d(\Cc)-1}-m_k+1.$$
Then 
$$d_r(\Cc)=\cdots=d_{r+m_k-1}(\Cc)=h.$$
\end{theorem}

\begin{remark}
One can also write down the generalized weights computed in Theorem \ref{genwtsMSRD} as follows.
Let $j\in[\ell]$, $0\leq \delta\leq n_j-1$, and let $\Cc\subseteq\MM$ be an MSRD code with $d(\Cc)=\sum_{i=1}^{j-1} n_i+\delta+1$ and $\dim(\Cc)=\sum_{i=j}^{\ell} m_in_i-\delta m_j$. Write
$$h=\sum_{i=1}^{k-1}n_i+\varepsilon+1$$ where $k\geq j$. Since $d(\Cc)\leq h\leq n$, one has that
$\delta\leq\varepsilon\leq n_j-1$ if $k=j$, and $0\leq \varepsilon\leq n_k-1$ if $k>j$. Then
$$r=\left\{\begin{array}{cc}
    (\varepsilon-\delta)m_j+1 & \mbox{ if } k=j,\, \delta\leq \varepsilon\leq n_j-1,\\
    (n_j-\delta)m_j+\sum_{i=j+1}^{k-1} m_in_i+\varepsilon m_k+1 & \mbox{ if } j< k\leq\ell,\, 0\leq \varepsilon\leq n_k-1.
\end{array}\right.$$
\end{remark}

\begin{remark}
It follows from Theorem \ref{genwtsMSRD} that both bounds in the statement of Theorem \ref{singletonbound} are met for $r\in[\dim(\Cc)]$ of the form $r=1,m_j+1,\ldots,(n_j-\delta-1)m_j+1$, and $$r=(n_j-\delta)m_j+\sum_{i=j+1}^{k-1} m_in_i+\varepsilon m_k+1$$ with $j<k\leq \ell$ and $0\leq \varepsilon\leq n_k-1$.
\end{remark}

\begin{remark}\label{rmk:genwtsintervals}
Let $d_0(\Cc)=0$ and $d_{\dim(\Cc)+1}(\Cc)=n+1$. Theorem \ref{genwtsMSRD} states that, for any $d(\Cc)\leq h\leq n$ and $r$ of the form $r=r_h-r_{d(\Cc)-1}-m_k+1$, we have
$$d_{r-1}(\Cc)<d_{r}(\Cc)=\ldots=d_{r+m_k-1}(\Cc)<d_{r+m_k}(\Cc).$$
\end{remark}

Inspired by Remark \ref{rmk:genwtsintervals} and by the definition of $r$-MRD codes, we define a notion of $r$-MSRD code as follows.
Notice that being $1$-MSRD is equivalent to being MSRD.

\begin{definition}\label{defn:rMSRD}
Let $j\in[\ell]$, $0\leq \delta\leq n_j-1$, and let $\Cc\subseteq\MM$ be a code of $\dim(\Cc)=\sum_{i=j}^\ell m_in_i-\delta m_j$. 
Define $d_{\max}=\sum_{i=1}^{j-1}n_i+\delta+1$, let $d_{\max}\leq h\leq n$ and $$r=r_h-r_{d_{\max}-1}-m_k+1,$$ 
where $k=\max\{\nu\mid \sum_{i=1}^{\nu-1}n_i<h\}$.
We say that $\Cc$ is {\bf $r$-MSRD} if 
$$d_r(\Cc)=h.$$
\end{definition}

We conclude this section by showing that, if $\Cc$ is $r$-MSRD, then $\Cc$ is $r'$-MSRD for all $r'\geq r$, where $r,r'$ are integers of the form given in Definition \ref{defn:rMSRD}. This observation allows us to compute the generalized weights of an $r$-MSRD code. 
Since an MSRD code is $1$-MSRD, the proof of next theorem also proves Theorem \ref{genwtsMSRD}.

\begin{theorem}\label{thm:rMSRD}
Let $j\in[\ell]$, $0\leq\delta\leq n_j-1$, and let $\Cc\subseteq\MM$ be a non-trivial code of $\dim(\Cc)=\sum_{i=j}^\ell m_in_i-\delta m_j$. Define $d_{\max}=\sum_{i=1}^{j-1}n_i+\delta+1$, let $d_{\max}\leq h\leq n$ and $$r=r_h-r_{d_{\max}-1}-m_k+1,$$ 
where $k=\max\{\nu\mid \sum_{i=1}^{\nu-1}n_i<h\}$. If $\Cc$ is $r$-MSRD, then  $$d_{r}(\Cc)=\ldots=d_{r+m_k-1}(\Cc)=h.$$
Moreover, $\Cc$ is $(r+m_k)$-MSRD. 
\end{theorem}

\begin{proof}
We have $$h=d_{r}(\Cc)\leq\ldots\leq d_{r+m_k-1}(\Cc)\leq h,$$
where the equality follows from the definition of $r$-MSRD code, the first and second inequalities from Proposition \ref{properties}, and the third from Corollary \ref{cor.rsb}. Therefore $d_{r}(\Cc)=\ldots=d_{r+m_k-1}(\Cc)=h$.

Since $d_{r+m_k}(\Cc)\geq d_r(\Cc)=h>\sum_{i=1}^{k-1}n_i$, then by Lemma \ref{lemma}
$$d_{r+m_k}(\Cc)\geq d_r(\Cc)+1=h+1.$$ The reverse inequality follows from Corollary \ref{cor.rsb}, hence $d_{r+m_k}(\Cc)=h+1$.
Since $$\max\left\{\nu : \sum_{i=1}^{\nu-1}n_i<h+1 \right\}=\left\{\begin{array}{ll}
    k & \mbox{ if } \varepsilon<n_k-1, \\
    k+1 & \mbox{ if } \varepsilon=n_k-1,
\end{array}\right.$$ we let
$$m'=\left\{\begin{array}{ll}
    m_k & \mbox{ if } \varepsilon<n_k-1,\\
    m_{k+1} & \mbox{ if } \mbox{ if } \varepsilon=n_k-1.
\end{array}\right.$$
Since $m'=r_{h+1}=r_h$, one has that $r+m_k=r_{h+1}-r_{d_{\max}-1}-m'+1$, hence we proved that $\Cc$ is $(r+m_k)$-MSRD.
\end{proof}

\begin{remark}
We follow the notation of the last theorem. If a code $\Cc$ is such that $d_r<h$ but $d_{r+s}(\Cc)=h$ for some $1\leq s\leq m_k-1$ then by Corollary \ref{cor.rsb} we have    $$d_{r+s}(\Cc)=\cdots=d_{r+m_k-1}(\Cc)=h.$$
However, this does not imply that $\Cc$ is an $(r+m_k)$-MSRD code, as the next example shows.
\end{remark}

\begin{example}
An MSRD code $\mathcal{D}$ of dimension 4 in $\mathbb{F}_2^{4\times 4}\times\mathbb{F}_2^{4\times 2}\times\mathbb{F}_2^{2\times 2}$ has weights $d_1(\mathcal{D})=d_2(\mathcal{D})=7$, $d_3(\mathcal{D})=d_4(\mathcal{D})=8$. Let $\mathcal{C}$ be generated by
\begin{align*}
&\left\{\left(\begin{pmatrix}1&0&0&0\\ 0&1&0&0\\0&0&1&0\\0&0&0&1\end{pmatrix},\begin{pmatrix} 1&0\\ 0&1\\ 0&0\\ 0&0\end{pmatrix},\begin{pmatrix} 1&0\\ 0&0\end{pmatrix}\right)\right. ,\\
&\left(\begin{pmatrix}0&0&0&1\\1&0&0&1\\ 0&1&0&0\\ 0&0&1&0\end{pmatrix},\begin{pmatrix}0&1\\ 1&1\\ 0&0\\ 0&0\end{pmatrix},\begin{pmatrix} 0&0\\ 1&0\end{pmatrix}\right),\\
&\left(\begin{pmatrix}0&0&1&0\\0&0&1&1\\1&0&0&1\\0&1&0&0\end{pmatrix},\begin{pmatrix} 0&0\\ 0&0\\ 1&0\\ 0&1\end{pmatrix},0\right),\\
&\left.\left(0,0,\begin{pmatrix}0&0\\0&1\end{pmatrix}\right)\right\}
\end{align*}
The code $\mathcal{C}$ has dimension $4$ and $d_1(\mathcal{C})=1$, then $\mathcal{C}$ is not MSRD. We checked using the computer algebra system Macaulay2 \cite{M2} that the only nonzero codewords of $\Cc$ of sum-rank less than 7 are the third and the fourth element in the previous list. Hence $d_2(\mathcal{C})=d_2(\mathcal{D})=7$. Taking $\mathcal{A}=\mathbb{F}_2^{4\times 4}\times\mathbb{F}_2^{4\times 2}\times\left\{\begin{pmatrix} a&0\\ b&0\end{pmatrix}\ |\ a,b\in\mathbb{F}_2\right\}$ we can see that $d_3(\mathcal{C})=7<8=d_3(\mathcal{D})$. In particular, $\mathcal{C}$ is not $3$-MSRD.
\end{example}

\section*{Appendix: Support spaces and information leakage}

An alternative notion of generalized weights could be defined for linear codes
$$ \mathcal{C} \subseteq \mathbb{F}_q^{m_1 \times n_1} \times \mathbb{F}_q^{m_2 \times n_2} \times \cdots \times \mathbb{F}_q^{m_\ell \times n_\ell} $$
as follows. Here, we do not assume that $ n_i \leq m_i $, for $ i = 1,2, \ldots, \ell $. For positive integers $ m $ and $ n $, we say that $ \mathcal{V}_\mathcal{L} \subseteq \mathbb{F}_q^{m \times n} $ is the (row) support space associated to the vector space $ \mathcal{L} \subseteq \mathbb{F}_q^n $ if
$$ \mathcal{V}_\mathcal{L} = \{ C \in \mathbb{F}_q^{m \times n} : {\rm Row}(C) \subseteq \mathcal{L} \}, $$
where $ {\rm Row}(C) \subseteq \mathbb{F}_q^n $ denotes the row space of $ C \in \mathbb{F}_q^{m \times n} $. Clearly $ \mathcal{V}_\mathcal{L} $ is an $ \mathbb{F}_q $-linear subspace of $ \mathbb{F}_q^{m \times n} $ of dimension
$$ \dim(\mathcal{V}_\mathcal{L}) = m \dim(\mathcal{L}). $$
Denote by $ \mathcal{P}_{q,n} $ the collection of subspaces of $ \mathbb{F}_q^n $. For a linear code $ \mathcal{C} \subseteq \mathbb{F}_q^{m_1 \times n_1} \times \mathbb{F}_q^{m_2 \times n_2} \times \cdots \times \mathbb{F}_q^{m_\ell \times n_\ell} $, we may give an alternative definition of generalized weights as
\begin{equation}
\begin{split}
d^{Supp}_r(\Cc) = \min \Bigg\lbrace \sum_{i=1}^\ell \dim(\mathcal{L}_i) :\, & \mathcal{L}_i \subseteq \mathcal{P}_{q,n_i}, 1 \leq i \leq \ell , \\ 
& \dim \left( \Cc \cap \left( \mathcal{V}_{\mathcal{L}_1} \times \cdots \times \mathcal{V}_{\mathcal{L}_\ell} \right) \right) \geq r \Bigg\rbrace .
\end{split}
\label{eq gen weights based on supports}
\end{equation}
for $r\in[\dim(\mathcal{C})]$. In the case where,for all $i\in[\ell]$, $ n_i < m_i $ or $ n_i = m_i = 1 $, then by \cite[Theorem 26]{Rav16} 
$$ d^{Supp}_r(\Cc) = d_r(\Cc), $$
for all linear codes $ \Cc \subseteq \mathbb{F}_q^{m_1 \times n_1} \times \mathbb{F}_q^{m_2 \times n_2} \times \cdots \times \mathbb{F}_q^{m_\ell \times n_\ell} $ and all $r\in[\dim(\Cc)]$. In particular, both coincide with the generalized Hamming weights if $ m_i = n_i = 1 $ for $i\in[\ell]$.

Consider now arbitrary values of $ m_i $ and $ n_i $, for $i\in[\ell]$. The weights in (\ref{eq gen weights based on supports}) present an advantage and a disadvantage with respect to using anticodes instead of support spaces. Their disadvantage is that such weights are not always invariant by arbitrary linear sum-rank isometries (simply notice that support spaces are not necessarily again support spaces after transposition of matrices, as we are only considering row supports). On the other hand, their advantage is that they measure information leakage to a wire-tapper in scenarios such as multishot linear network coding \cite{secure-multishot}. 

More concretely, consider a linear code $ \Cc \subseteq \mathbb{F}_q^{m_1 \times n_1} \times \mathbb{F}_q^{m_2 \times n_2} \times \cdots \times \mathbb{F}_q^{m_\ell \times n_\ell} $. Choose a complementary vector space $ \mathcal{M} \oplus \Cc = \mathbb{F}_q^{m_1 \times n_1} \times \mathbb{F}_q^{m_2 \times n_2} \times \cdots \times \mathbb{F}_q^{m_\ell \times n_\ell} $. We may see $ \mathcal{M} $ as our space of messages \cite[Definition 3]{secure-multishot}. A (random) message $ M \in \mathcal{M} $ is encoded by choosing $ C \in \Cc $ uniformly at random in $ \Cc $, and finally setting $ D = M + C $, in order to hide $ M $. If we set $ D = (D_1, D_2, \ldots, D_\ell) $, then $ D_i \in \mathbb{F}_q^{m_i \times n_i} $ is sent through an $ \mathbb{F}_q $-linearly coded network with $ n_i $ outgoing links from the source node, for $ i\in[\ell] $. This could be the scenario in multishot network coding without delays (or treating delays as erasures), or in singleshot network coding where we know that the network has at least $ \ell $ disconnected components (after removing the source and sink nodes), see \cite{secure-multishot}. Assume we have no further knowledge of the network topology or linear network code, but we know that an adversary wire-taps $ \mu $ arbitrary links of the network. Then the information contained in
$$ W = (D_1 B_1, D_2 B_2, \ldots, D_\ell B_\ell) $$
is leaked to the wire-tapper, for matrices $ B_i \in \mathbb{F}_q^{n_i \times \mu_i} $, for $ i\in[\ell] $, such that 
$$ \sum_{i=1}^\ell {\rk}(B_i) \leq \sum_{i=1}^\ell \mu_i = \mu.$$
Equality may be attained, as we do not know nor have control over the matrices $ B_i $. By \cite[Lemma 1]{secure-multishot}, the amount of information on $ M $ leaked to the wire-tapper, measured in symbols in $ \mathbb{F}_q $, is given by the mutual information in base $ q $
$$ I_q(M ; W) \leq \dim \left( \Cc^\perp \cap ( \mathcal{V}_{\mathcal{L}_1} \times \cdots \times \mathcal{V}_{\mathcal{L}_\ell} ) \right), $$
where $ \mathcal{L}_i = {\rm Row}(B_i^t) \in \mathcal{P}_{q,n_i} $, for $ i\in[\ell] $. Furthermore, equality holds if $ M $ is chosen uniformly at random in $ \mathcal{M} $. Thus, in such a situation, $ d^{Supp}_r(\Cc^\perp) $ represents the minimum number of links that an adversary needs to wire-tap in order to obtain at least $ r $ units of information (number of bits multiplied by $ \log_2 (q) $) of the sent message.

We conclude with some remarks on the properties of the generalized weights $ d^{Supp}_r(\Cc) $ as above. First, if $ m = m_1 = m_2 = \ldots = m_\ell $, $ n = n_1 + n_2 + \cdots + n_\ell $, and $ \Cc \subseteq \mathbb{F}_{q^m}^n \cong \mathbb{F}_q^{m_1 \times n_1} \times \mathbb{F}_q^{m_2 \times n_2} \times \cdots \times \mathbb{F}_q^{m_\ell \times n_\ell} $ is an $ \mathbb{F}_{q^m} $-linear code, then 
$$ d^{Supp}_{rm - t}(\Cc) = d^{SR}_r(\Cc), $$
for $ t\in\{0,\ldots, m-1\} $ and $r\in[\dim_{\mathbb{F}_{q^m}}(\Cc)]$, where $ d^{SR}_r(\Cc) $ denotes the $ r $th generalized weight considered in \cite[Definition 10]{gsrws}. This equality is easy to prove and recovers \cite[Theorem 7]{rgmw} when $\ell = 1$. Moreover, this is the analogue of Theorem \ref{thm:klinear} for the case $m =m_1=m_2=\ldots=m_\ell$. Notice that the assumption that $n_i<m_i$ is not needed in this setting.

Finally, if we assume that $ m_1 \geq m_2 \geq \ldots \geq m_\ell $, then all of the properties stated throughout this manuscript for the generalized weights $ d_r(\Cc) $ hold mutatis mutandis for the generalized weights $ d^{Supp}_r(\Cc) $, with similar proofs and without assuming that $ n_i \leq m_i $ for $ i\in[\ell]$. In order to prove this claim, it suffices to prove the analogue of Proposition \ref{properties}, Theorem \ref{singletonbound} and Lemma \ref{lemma} for the generalized weights $d^{Supp}_r(\Cc)$. In fact, the other properties follow from those results. We start with an analogue of Proposition \ref{properties}.

\begin{proposition}\label{properties_app} 
Let $0\neq\Cc \subseteq\mathcal{D} \subseteq \MM$, then:
\begin{enumerate}
\item $d^{Supp}_1(\Cc)= d(\Cc)$,
\item $d^{Supp}_r(\Cc)\leq d^{Supp}_s(\Cc)$ for $1\leq r\leq s\leq\dim(\Cc)$,
\item $d^{Supp}_r(\Cc)\geq d^{Supp}_r(\mathcal{D})$ for $r\in[\dim(\Cc)]$,
\item $d^{Supp}_{\dim(\Cc)}(\Cc)\leq n_1 + \dots + n_\ell$,
\item $d^{Supp}_{r+n_1m_1+\cdots+n_{j-1}m_{j-1}+\delta m_{j}}(\Cc)\geq d^{Supp}_r(\Cc)+n_1+\cdots+n_{j-1}+\delta$\\ 
for $j\in[\ell]$, $r\in[ \dim(\Cc)-(n_1m_1+\cdots+n_{j-1}m_{j-1}+\delta m_j)]$, and $0\leq\delta\leq n_j-1$.
\end{enumerate}
\end{proposition}

\begin{proof} 
1. Let $C=(C_1,\ldots,C_\ell)\in \Cc$ be an element of minimum sum-rank. Let $\mathcal{L}_i={\rm Row}(C_i)$ be a subspace such that $C_i\in \mathcal{V}_{\mathcal{L}_i}$ and $\dim(\mathcal{V}_{\mathcal{L}_i}) = m_i \rk(C_i)$, for $i\in[ \ell ]$. Then $\Cc \cap (\mathcal{V}_{\mathcal{L}_1} \times \cdots \times \mathcal{V}_{\mathcal{L}_\ell})\neq 0$, hence $d^{Supp}_1(\Cc)\leq d(\Cc)$. To prove that they are equal, we observe that if $ \mathcal{L}_i^\prime \in \mathcal{P}_{q, n_i} $, for $ i \in [\ell] $, are such that $ \sum_{i=1}^\ell \dim(\mathcal{L}^\prime_i) < d(\Cc) $, then $\Cc \cap (\mathcal{V}_{\mathcal{L}^\prime_1} \times \cdots \times \mathcal{V}_{\mathcal{L}^\prime_\ell}) = 0$.

2., 3., and 4. follow directly from the definition.

5. Let $s=r+n_1m_1+\cdots+n_{j-1}m_{j-1}+\delta m_j$.
Let $ \mathcal{L}_i \in \mathcal{P}_{q,n_i} $, for $ i \in [\ell] $, such that $ \dim(\Cc \cap (\mathcal{V}_{\mathcal{L}_1} \times \cdots \times \mathcal{V}_{\mathcal{L}_\ell})) \geq s $ and $d^{Supp}_s(\Cc)= \sum_{i=1}^\ell \dim(\mathcal{L}_i)$. For $i\in[\ell]$, let $u_i=\dim(\mathcal{L}_i)$. Since $$\sum_{i=1}^{\ell} m_iu_i=\dim(\mathcal{V}_{\mathcal{L}_1} \times \cdots \times \mathcal{V}_{\mathcal{L}_\ell})\geq \dim(\Cc \cap (\mathcal{V}_{\mathcal{L}_1} \times \cdots \times \mathcal{V}_{\mathcal{L}_\ell})) $$
$$\geq s> n_1m_1+\cdots+n_{j-1}m_{j-1}+\delta m_j$$ and $m_1\geq\cdots\geq m_\ell$, then $d^{Supp}_s(\Cc)=u_1+\ldots+u_\ell> n_1+\cdots+n_{j-1}+\delta$ by Lemma~\ref{lemma:nos}. Let $v_1,\ldots,v_{\ell}$ be such that $n_1+\cdots+n_{j-1}+\delta= v_1+\cdots+v_\ell$ and $v_i\leq u_i$ for $i\in[\ell]$.
We have that $n_1m_1+\cdots+n_{j-1}m_{j-1}+\delta m_j\geq v_1m_1+\cdots+v_\ell m_\ell$, since $m_1\geq\cdots\geq m_\ell$. For all $i\in[\ell]$ there exist vector spaces $\mathcal{L}'_i\subseteq \mathcal{L}_i$ of $\dim(\mathcal{L}'_i)=u_i-v_i$. Then
\begin{align*}
\dim(\Cc \cap (\mathcal{V}_{\mathcal{L}'_1} \times \cdots \times \mathcal{V}_{\mathcal{L}'_\ell}))&\geq s-(v_1m_1+\cdots+v_\ell m_ \ell)\\
&\geq s-(n_1m_1+\cdots+n_{j-1}m_{j-1}+\delta m_j)\\
&=r
\end{align*}
hence
\begin{equation*}
d_r(\Cc)\leq \sum_{i=1}^\ell (u_i-v_i) = d_s(\Cc)-(n_1+\cdots+n_{j-1}+\delta).\qedhere
\end{equation*}
\end{proof}

The Singleton-type bound from Corollary \ref{cor:singletontypebound} is a direct consequence of Proposition \ref{properties}. Notice however that it also follows directly from Proposition \ref{properties_app} for the generalized weights $d^{Supp}_r(\Cc)$. 
The next theorem is an analogue of Theorem \ref{singletonbound} for the generalized weights $d^{Supp}_r(\Cc)$.

\begin{theorem}\label{singletonbound_app}
Let $\Cc\subseteq\MM$ be a code and let $r\in[\dim(\Cc)]$. Let $j\in[\ell]$ and $0\leq\delta\leq n_j-1$ be such that $$d^{Supp}_r(\Cc)-1\geq\sum_{i=1}^{j-1} n_i+\delta.$$
Then
\begin{equation}
\dim(\Cc)\leq\sum_{i=j}^\ell m_in_i-m_{j}\delta+r-1.
\end{equation}
\end{theorem}

\begin{proof}
Let $\mathcal{L}_i=\mathbb{F}_q^{n_i}$ for $i\in[j-1]$, let $\mathcal{L}_j \in \mathcal{P}_{q,n_j} $ be of dimension $\delta $, and let $ \mathcal{L}_i=0$ for $j+1\leq i\leq\ell$. Then
$$\dim(\Cc\cap (\mathcal{V}_{\mathcal{L}_1} \times \cdots \times \mathcal{V}_{\mathcal{L}_\ell}))\leq r-1.$$
Therefore
\begin{align*}
\dim(\Cc)+\sum_{i=1}^{j-1} m_in_i+m_j \delta-r+1&\leq\dim(\Cc)+\dim(\mathcal{V}_{\mathcal{L}_1} \times \cdots \times \mathcal{V}_{\mathcal{L}_\ell})\\
&-\dim(\Cc\cap (\mathcal{V}_{\mathcal{L}_1} \times \cdots \times \mathcal{V}_{\mathcal{L}_\ell}))\\
&=\dim(\Cc+(\mathcal{V}_{\mathcal{L}_1} \times \cdots \times \mathcal{V}_{\mathcal{L}_\ell}))\leq\sum_{i=1}^\ell m_in_i.
\end{align*}
\end{proof}

A version of the Singleton Bound for the generalized weights $d^{Supp}_r(\Cc)$ follows directly from Theorem \ref{singletonbound_app} and yields the same inequalities as in Corollary \ref{cor.rsb}.
Finally, we establish the analogue of Lemma \ref{lemma} for the generalized weights $ d^{Supp}_r(\Cc) $.

\begin{lemma}\label{lemma_app}
Let $\Cc\subseteq\MM$ be a code and let $k\in[\ell]$, $r+m_k\in[\dim(\Cc)]$. If
$$d^{Supp}_{r+m_k}(\Cc)>\sum_{i=1}^{k-1}n_i$$ then
$$d^{Supp}_{r+m_k}(\Cc)\geq d^{Supp}_r(\Cc)+1.$$
\end{lemma}

\begin{proof}
For $i\in [\ell]$, let $\mathcal{L}_i \in \mathcal{P}_{q,n_i}$ be such that $d^{Supp}_{r+m_k}(\Cc)=\sum_{i=1}^\ell \dim(\mathcal{L}_i)$ and $\dim(\Cc\cap(\mathcal{V}_{\mathcal{L}_1}\times\cdots\times
\mathcal{V}_{\mathcal{L}_\ell}))\geq r+m_k$. We claim that there exists $k\leq j\leq \ell$ such that $\mathcal{L}_j\neq 0$. In fact, if this were not the case, then 
$$\sum_{i=1}^{k-1}n_i\geq \sum_{i=1}^\ell \dim(\mathcal{L}_i) =d^{Supp}_{r+m_k}(\Cc).$$
For $i\in[\ell]$, let $\mathcal{L}^\prime_i \in\mathcal{P}_{q,n_i}$ be such that $\mathcal{L}^\prime_i = \mathcal{L}_i$ if $ i \neq j $, and $\mathcal{L}^\prime_j \subseteq \mathcal{L}_j$ with $\dim(\mathcal{L}^\prime_j)=\dim(\mathcal{L}_j)-1$.
One has
$$\dim(\Cc\cap(\mathcal{V}_{\mathcal{L}^\prime_1} \times \cdots \times \mathcal{V}_{\mathcal{L}^\prime_\ell}))\geq\dim(\Cc\cap(\mathcal{V}_{\mathcal{L}_1} \times \cdots \times \mathcal{V}_{\mathcal{L}_\ell}))-m_j\geq r+m_k-m_j\geq r,$$
hence
$$d^{Supp}_r(\Cc)\leq \sum_{i=1}^\ell \dim(\mathcal{L}^\prime_i) =d^{Supp}_{r+m_k}(\Cc)-1.$$
\end{proof}


We conclude by noting that Wei's duality (Theorem \ref{weiduality}) also holds for the generalized weights $ d^{Supp}_r(\Cc) $ with the same proof.

\bibliographystyle{plain}

\end{document}